\documentclass[a4paper,USenglish,cleveref, autoref, thm-restate]{lipics-v2021}

\usepackage{colortbl}
\usepackage{tikz}
\newcommand{\morph}{\mu}
\newcommand{\tmmorph}{\tau}
\newcommand{\fibmorph}{\varphi}
\newcommand{\pdmorph}{\pi}
\newcommand{\bwas}{AS_\morph}
\newcommand{\bwms}{MS_\morph}

\newcommand{\bwt}{\textsf{bwt}}
\newcommand{\R}[1]{\mathcal{R}(#1)}
\newcommand{\La}{\mathcal{L}}

\newcommand{\rbwt}{r}
\newcommand{\lcp}{lcp}
\newcommand{\lcs}{lcs}
\def\rle(#1){\textsf{rle}(#1)}
\newcommand{\conjclass}{\mathcal{R}}

\newcommand{\fact}{\mathcal{F}}
\newcommand{\cfact}{\widetilde{\mathcal{F}}}

\newcommand\red[1]{{\color{red} #1}}

\newcommand{\Prim}[1]{Q(#1)}
\newcommand{\Pow}[1]{\overline{Q(#1)}}

\newcommand{\Fib}{F}
\newcommand{\vir}[1]{``#1''}
\usepackage{arydshln}
\newcommand\Comment[1]{}
\usepackage{comment}

\title{Morphisms and BWT-run Sensitivity} %
\titlerunning{Morphisms and BWT-run Sensitivity} 

\author{Gabriele Fici}{Department of Mathematics and Computer Science, University of Palermo, Italy}{gabriele.fici@unipa.it}{https://orcid.org/0000-0002-3536-327X}{Supported by MUR project PRIN 2022 APML – 20229BCXNW, funded by the European Union – Mission 4 "Education and Research" C2 - Investment 1.1.}

\author{Giuseppe Romana}{Department of Mathematics and Computer Science, University of Palermo, Italy}{giuseppe.romana01@unipa.it}{https://orcid.org/0000-0002-3489-0684}{Supported by the MUR PRIN Project \vir{PINC, Pangenome INformatiCs: from Theory to Applications} (Grant No.\ 2022YRB97K), funded by Next Generation EU PNRR M4 C2, Inv. 1.1.}

\author{Marinella Sciortino}{Department of Mathematics and Computer Science, University of Palermo, Italy}{marinella.sciortino@unipa.it}{https://orcid.org/0000-0001-6928-0168}{Supported by the project \vir{ACoMPA – Algorithmic and Combinatorial Methods for Pangenome Analysis} (CUP B73C24001050001) funded by the NextGeneration EU programme PNRR ECS00000017 Tuscany Health Ecosystem (Spoke 6), Mission 4, Component 2.}

\author{Cristian Urbina}{Department of Computer Science, University of Chile, Santiago, Chile \and Centre for Biotechnology and Bioengineering (CeBiB), Santiago, Chile}{crurbina@dcc.uchile.cl}{https://orcid.org/0000-0001-8979-9055}{Basal Funds FB0001 and AFB240001; Fondecyt Grant 1-230755, ANID, Chile; ANID-Subdirección de Capital Humano/Doctorado Nacional/2021-21210580, ANID, Chile.}

\authorrunning{G. Fici, G. Romana, M. Sciortino, C. Urbina} 

\Copyright{Gabriele Fici, Giuseppe Romana, Marinella Sciortino and Cristian Urbina} 



\category{} 

\relatedversion{} 

\nolinenumbers 

\EventEditors{John Q. Open and Joan R. Access}
\EventNoEds{2}
\EventLongTitle{42nd Conference on Very Important Topics (CVIT 2016)}
\EventShortTitle{CVIT 2016}
\EventAcronym{CVIT}
\EventYear{2016}
\EventDate{December 24--27, 2016}
\EventLocation{Little Whinging, United Kingdom}
\EventLogo{}
\SeriesVolume{42}
\ArticleNo{23}

\begin{document}

\maketitle

\begin{abstract}
We study how the application of injective morphisms affects the number $r$ of equal-letter runs in the Burrows–Wheeler Transform (BWT). This parameter has emerged as a key repetitiveness measure in compressed indexing. We focus on the notion of BWT-run sensitivity after application of an injective morphism. For binary alphabets, we characterize the class of morphisms that preserve the number of BWT-runs up to a bounded additive increase, by showing that it coincides with the known class of primitivity-preserving morphisms, which are those that map primitive words to primitive words. We further prove that deciding whether a given binary morphism has bounded BWT-run sensitivity is possible in polynomial time with respect to the total length of the images of the two letters. Additionally, we explore new structural and combinatorial properties of synchronizing and recognizable morphisms. These results establish new connections between BWT-based compressibility, code theory, and symbolic dynamics. 
\end{abstract}

\newpage

\section{Introduction}
Morphisms are a powerful combinatorial mechanism for generating a collection of repetitive texts, and have been largely used in the field of combinatorics on words and formal languages \cite{lothaire1997combinatorics, Rigo2014}. Formally, a morphism maps each character of an alphabet to a word over the same or another alphabet, by preserving the operation of concatenation. That is, if $\morph$ is a morphism and $u$ and $v$ are words, then $\morph(uv)=\morph(u)\morph(v)$. Iterating morphisms can produce long and often highly repetitive sequences, which makes them a natural model for studying repetitiveness in words. Morphisms find applications in a wide range of contexts. Injective morphisms are widely used in information theory, data compression, and cryptography, as they define uniquely decodable codes~\cite{codesautomata}. More recently, morphisms have been employed in combination with copy-paste mechanisms to define novel compression schemes, known as NU-systems~\cite{NAVARRO_Urbina_TCS2025}, further highlighting their versatility in modeling and processing repetitive data.

The Burrows--Wheeler Transform (BWT) is a reversible transformation introduced in 1994 in the field of data compression \cite{BW94} and now underpins some of the most used tools in bioinformatics, such as {\tt bwa}~\cite{bwa,VasimuddinMLA19} and {\tt bowtie}~\cite{Bowtie,bowtie2}. It permutes the characters of a text in a way that makes it more compressible, by clustering characters that precede similar contexts in the text. This property often results in long runs of identical characters, particularly in repetitive texts. The number $r$ of such equal-letter runs, known as \emph{BWT-runs}, has recently emerged as a measure of repetitiveness \cite{Navarro22cacm}. Several measures have been proposed to quantify repetitiveness in strings \cite{Navacmcs20.3}, such as the number $z$ of phrases in the Lempel--Ziv parsing, the size $g$ of the smallest context-free grammar generating the text, the size $\gamma$ of the smallest string attractor \cite{KempaP18, RomanaJCTA}. Among these, the measure $r$ has recently attracted considerable attention due to its close connection with compressed indexing structures, such as the $r$-index \cite{rindex}, which use space proportional to $r$ and support efficient pattern matching and retrieval in highly repetitive text collections, including genomic datasets and versioned document archives. Akagi et al.~\cite{AKAGI2023} explored the question of how much one character edit affects compression-based repetitiveness measures. In \cite{Giuliani2025}, the effect of single edit operations on the measure $r$ has also been analyzed.

In this paper, we study how the application of an injective morphism affects the measure $r$, i.e., the number of BWT-runs. We focus on two notions of  \emph{BWT-run sensitivity}, which capture how much the number of BWT-runs can change when a morphism $\morph$ is applied to a word.  The \emph{additive sensitivity function} $AS_{\morph}$ gives, for every $n>0$, the maximum increase in the number of BWT-runs that can occur when applying the morphism $\morph$ to any word of length $n$, while the \emph{multiplicative sensitivity function} $MS_{\morph}$ gives, for every $n>0$, the maximum ratio between the number of BWT-runs after and before the morphism $\morph$ is applied, over all words of length $n$. These notions allow us to quantify the impact of a morphism on the compressibility of the resulting text. An initial approach to the study of how morphisms affect the number of BWT-runs was given in~\cite{Fici23}, where we showed that Sturmian morphisms are the only binary injective morphisms that preserve the number of BWT-runs. Here, we tackle the problem of characterizing those binary injective morphisms that preserve the BWT-based compressibility of a text, in the sense that they have an additive sensitivity function bounded by a constant.  We prove that this class coincides with the known class of \emph{primitivity-preserving morphisms}, which are those that map primitive words to primitive words. As a direct consequence, for these morphisms the multiplicative sensitivity function is also bounded.
Primitivity-preserving morphisms are a well-studied class in algebraic theory of codes, and they are crucial in applications involving symbolic sequences, code synchronization, and the structural analysis of words \cite{ShyrThierrin1997, Restivo74, domosi_93, Mitrana97, lothaire1997combinatorics, codesautomata, HolubRS23}. 

In addition to establishing a novel connection between BWT-based compressed indexing, combinatorics on words, and code theory, a key contribution of our paper consists in identifying new combinatorial and structural properties of primitivity-preserving, recognizable, and synchronizing morphisms. These properties are central to our main results but also hold independent interest in information theory and symbolic dynamics, where such morphisms play a fundamental role in coding, synchronization, and symbolic representations of dynamical systems \cite{BealPR_EJC24,codesautomata}. In fact, recognizability ensures that the morphic image of a word can be uniquely decomposed, up to rotations, into a sequence of morphic images of the letters of the alphabet. Synchronizing morphisms guarantee that a window of bounded length is sufficient to detect boundaries between codewords, a property that is crucial for decoding and synchronization in data streams. 

We further show that all binary injective morphisms have bounded multiplicative sensitivity, but this result does not extend to alphabets with more than two symbols.

Furthermore, we prove that it is decidable in polynomial time whether the additive sensitivity function of a binary morphism is bounded by a constant, which makes our results practically applicable to the design of compression and indexing techniques that work directly on morphic encodings of highly repetitive text collections. Such a result builds upon fundamental results in the field of combinatorics on words, including properties of codes and solutions to word equations. 

The rest of the paper is organized as follows. In Section \ref{sec:preliminaries}, we present the preliminaries on words, morphisms, and the BWT. Section \ref{sec:new_comb_morphisms} introduces new combinatorial and structural properties of primitivity-preserving, recognizable, and synchronizing morphisms. Section \ref{sec:sensitivity} formalizes the sensitivity of $r$ with respect to the application of morphisms and motivates our measures. Section \ref{sec:characterization_morphisms} contains our main theorem characterizing the morphisms with a bounded additive sensitivity function. Section \ref{sec:multiplicative_sensitivity} discusses the multiplicative case, and Section \ref{sec:final} concludes with final remarks and open problems.


\section{Preliminaries}\label{sec:preliminaries}

\paragraph*{Basics}

Let $\Sigma = \{a_1,a_2,\dots,a_{\sigma}\}$ be a finite sorted set of \emph{letters} $a_1<a_2<\dots <a_\sigma$, which we call an \emph{alphabet}.
A \emph{finite word} $w = w[0]w[1]\cdots w[|w|-1]$ is any finite sequence of letters where $w[i] \in \Sigma$, for $i \in [0, |w|-1]$, and $|w|$ is the \emph{length} of the word.
The \emph{empty word}, denoted by $\varepsilon$, is the unique word of length $0$.
The set of all finite words (resp.~all non-empty words) over the alphabet $\Sigma$ is denoted by $\Sigma^*$ (resp.~$\Sigma^+$).
  For a letter $a_i\in\Sigma$, $|w|_{a_i}$ denotes the number of occurrences of $a_i$ in $w$. The vector $(|w|_{a_{1}},\ldots,|w|_{a_{\sigma}})$ is called the \emph{Parikh vector} of $w$.
  
 If $u = u[0]\cdots u[n-1]$ and $v = v[0]\cdots v[m-1]$ are words, the \emph{concatenation} $uv$ of $u$ and $v$ is $uv  = u[0] \cdots u[n-1] v[0] \cdots v[m-1]$.
 We let $\Pi^{k}_{i=1} w_i$ denote the concatenation of the words $w_1, w_2,\dots,w_k$ in that order, and $w^k$ the concatenation of the word $w$ with itself $k$ times.
 
 For any $1\leq i\leq j\leq |w|$, we use the notation $w[i,j]$ to denote the word $w[i]w[i+1]\cdots w[j]$, which we call a \emph{factor} of $w$.
 If $i>j$, then we assume $w[i,j]=\varepsilon$.
 We let $\fact(w)$ denote the set of all factors of $w$. For any $\La\subseteq\Sigma^*$, we write $\fact(\La)=\bigcup_{w\in\La}\fact(w)$.
 A factor of $w$ is \emph{proper} if it is different from $w$ itself.
 The factor $w[i,j]$ is called a \emph{prefix} when $i = 1$, and a \emph{suffix} when $j = n$. 
The \emph{longest common prefix} between two words $u$ and $v$ is the longest word that is a prefix of both words. The length of this word is denoted by $\lcp(u,v)$. The \emph{longest common suffix} and the associated function $\lcs$ are defined symmetrically.

  The \emph{run-length encoding} of a word $w$, denoted $\rle(w)$, is the sequence of pairs $(c_i, l_i$) with $c_i \in \Sigma$ and $l_i > 0$, such that $w = c_1^{l_1}c_2^{l_2}\cdots c_r^{l_r}$ and $c_i \neq c_{i+1}$ for every $i\in[1,r-1]$. The length $|\rle(w)|$ is the number of \emph{equal-letter runs} in $w$.
 
 A \emph{rotation}, or \emph{conjugate}, of the word $w=w[0]w[1]\cdots w[n-1]$ is a word of the form $w[i+1,n-1] w[0,i]$, for some $0\leq i< n$, obtained by shifting $i$ letters cyclically.
 We let $\R{w}$ denote the multiset of all the $|w|$ rotations of $w$.
 A word in $\cfact(w):=\fact(\R{w})$ is called a \emph{circular factor} of $w$.
 
A word $w$ is \emph{primitive} if for every word $u \in \Sigma^+$, $w = u^k$ implies $k = 1$; otherwise, $w$ is called \emph{non-primitive} (or \emph{a power}). A word of length $n$ is primitive if and only if it has exactly $n$ distinct rotations, i.e., if $\R{w}$ has all-distinct elements.
We let $\Prim{\Sigma^*}$ denote the set of all primitive words in $\Sigma^*$, and $\Pow{\Sigma^*}$  the set of all non-primitive words in $\Sigma^*$. 
We say two non-empty words $u,v$ \emph{commute} if $uv=vu$. This is equivalent to saying that both $uv$ and $vu$ are not primitive. 
In a well-known paper~\cite{lyndon}, Lyndon and Sch\"utzenberger established strong connections between primitive words and some equations in a free group. We report these classical results in the Appendix~\ref{ap:Lyndon}.


\paragraph*{Codes and morphisms}

A set $X\subseteq \Sigma^+$ is a \emph{code} if for all $m,\ell\geq 0$ and $u_i,v_j\in X$ with $i\in[1,\ell], j\in[1,m]$, the equation $u_1 u_2 \cdots u_\ell = v_1 v_2 \cdots v_m$ implies that $\ell = m$ and $u_k = v_k$, for all $k\in [1,\ell]$. Or equivalently, every word $w\in X^+$ has a unique factorization in words in $X$.
Given a word $w\in X^+$, a word $u$ is an \emph{$X$-factor} of $w$ if there exists a rotation $w'$ of $w$ (which can be $w$ itself) that can be factored as $w' = s u  p$ such that $u, ps\in X^*$.

Whenever a code $X$ consists of two words, the following property holds~\cite{Huang,Shyr77}.

\begin{lemma}
\label{le:codedontcommute}
    A set $X=\{u,v\}$, $u,v\in \Sigma^+$, is a code if and only if $u$ and $v$ do not commute, i.e., $uv\neq vu$.
\end{lemma}

If $u$ and $v$ do not commute, then they are not powers of the same word, but in principle this does not exclude the case that either $u$, $v$, or both are non-primitive. For example, $X=\{aa,bb\}$ is a code.

Let $\Sigma$ and $\Gamma$ be two alphabets. A \emph{morphism}  $\mu$ is a map from $\Sigma^*$ to  $\Gamma^*$ such that $\mu(uv) = \morph(u)\morph(v)$ for all words $u, v \in \Sigma^*$.
Therefore, a morphism $\mu$ can be defined by specifying its action on the letters of $\Sigma$, and can therefore be denoted as $\morph=(\mu(a_1),\ldots,\mu(a_\sigma))$. 
The \emph{size} of the morphism $\morph$ is defined as $|\morph| = \sum_{c\in\Sigma}|\morph(c)|$. 
When $\Sigma=\Gamma$, for all $t> 0$ and $w\in\Sigma^+$, we have $\morph^t(w)=\morph(\morph^{t-1}(w))$ and $\morph^0(w)=w$.

\begin{remark}\label{rem:rot}
Let $\morph$ be a morphism. If $w$ and $w'$ are conjugates, then so are  $\morph(w)$ and $\morph(w')$. Moreover, since every conjugate of a power is a power, if   $\morph(w)$ is a power, so is $\morph(w')$ for every conjugate $w'$ of $w$.
\end{remark}

A morphism $\mu$ is \emph{cyclic} if there exists $z\in \Gamma^+$ such that $\mu(a)\in z^*$ for each $a\in \Sigma$. Otherwise, it is called \emph{acyclic}.

As shown in the following proposition, there is a very strong relation between codes and injective morphisms.

\begin{proposition}[\cite{codesautomata}]
\label{prop:codesmorphism}
Let $X\subset \Gamma^*$ be a code. Then, any morphism $\morph: \Sigma^* \rightarrow \Gamma^*$ which induces a bijection of some alphabet $\Sigma$ onto $X$ is injective. Conversely, let $\morph : \Sigma^* \rightarrow \Gamma^*$ be an injective morphism. Then, $X = \morph(\Sigma)$ is a code.
\end{proposition}


By Lemma~\ref{le:codedontcommute} and Proposition~\ref{prop:codesmorphism}, one can easily derive that for a binary morphism $\morph:\{a,b\}^*\rightarrow\Gamma^*$, injectivity is equivalent to acyclicity, which in turn is equivalent to the condition $\morph(ab)\neq\morph(ba)$.

Examples of injective morphisms are the \emph{Fibonacci morphism}  $\fibmorph= (ab,a)$, the \emph{Thue--Morse morphism}  $\tmmorph=(ab,ba)$, and the \emph{period-doubling morphism} $\pdmorph=(ab,aa)$.

From the relationship between codes and morphisms, many properties of codes are reflected in the corresponding properties of injective morphisms. Combinatorial properties of injective morphisms are explored in Section 3. 

The Fibonacci morphism belongs to a wider class of morphisms called \emph{Sturmian morphisms}, strictly related to the well-known Sturmian words~\cite{DBLP:conf/mfcs/BerstelS93}. 
Sturmian morphisms can be defined as those that can be obtained by composition from: the Fibonacci morphism $\varphi$, the morphism $E=(b,a)$, and the morphism $\Tilde{\fibmorph}=(ba,a)$.

Let us suppose that both $\Sigma$ and $\Gamma$ are endowed with a total order relation that yields a lexicographic order, denoted by $<_{\Gamma}$ and $<_{\Sigma}$, respectively.
A morphism $\morph : \Sigma^* \rightarrow \Gamma^*$ is \emph{abelian order-preserving} if for every pair of distinct words $x,y\in \Sigma^*$ having the same Parikh vector, it holds that $x <_{\Sigma} y \iff \mu(x) <_{\Gamma} \mu(y)$. 
A morphism $\mu$ is \emph{abelian order-reversing} if for every pair of distinct words $x$ and $y$ having the same Parikh vector, it holds that $x <_{\Sigma} y \iff \mu(x) >_{\Gamma} \mu(y)$. We simply write $<$ whenever $\Sigma$ and $\Gamma$ are clear from the context.

When $\Sigma=\{a,b\}$, the following result holds.

\begin{lemma}[\cite{Fici23}]
\label{lem:abelian_order}
Let $\mu : \{a,b\}^* \mapsto \Gamma^*$ be an acyclic morphism.
Then $\mu$ is either abelian order-preserving or abelian order-reversing.
\end{lemma}

For our purposes, the fact that binary acyclic morphisms are either abelian order-preserving or abelian order-reversing is a crucial property, since it implies that they preserve or reverse the order on the set of rotations of any given binary word.

\paragraph*{Burrows--Wheeler transform}
The \emph{Burrows--Wheeler transform} (BWT) of a word $w$, denoted by $\bwt(w)$, is a permutation of the letters of $w$ obtained by sorting all the rotations of $w$ in ascending lexicographic order and then concatenating the last letter of each rotation.
The original word can be recovered if one stores the position where it appears in the list of sorted rotations.
Figure  \ref{fig:BWT} shows the sorted rotations of the word $w=\fibmorph^4(a)=abaababa$ and $bwt(w)=bbbaaaaa$.

We let $\rbwt(w)$ denote the number of equal-letter runs of $\bwt(w)$, i.e., $\rbwt(w)=|\rle(\bwt(w))|$. Such a value can be considered as a measure of the repetitiveness of $w$. In fact, if a word $w$ is highly repetitive, the number of equal-letter runs of its BWT tends to be small. From Figure \ref{fig:BWT}, one can see that $r(abaababa)=2$.



One can easily verify that for each word $v\in \R{w}$, $\bwt(v)=\bwt(w)$ and, consequently, $r(v)=r(w)$ and  $r(\morph(v))=r(\morph(w))$ for every morphism $\morph$.

Let $w$ be a non-primitive word, i.e., $w=z^p$, for some $z\in \Sigma^+$ and $p>1$. It is well known that if $\bwt(z)=a_1a_2\cdots a_{|z|}$, then $\bwt(w)=a_1^pa_2^p\cdots a_{|z|}^p$~\cite{DBLP:journals/ipl/MantaciRS03}. This implies that $r(w)=r(z)$. 

Some results proved in~\cite{DBLP:journals/ipl/MantaciRS03, DBLP:journals/tcs/Paquin09, DBLP:journals/tcs/Chuan99} establish a strong connection between the BWT and Sturmian morphisms, as synthesized in the following theorem.

\begin{theorem}
    \label{bwsturmian}
Let $w$ be a word over $\{a,b\}$ that is not a power of a single letter. Then the following are equivalent:
\begin{enumerate}
    \item\label{morph0} $w=(\mu(a))^\ell$ for a Sturmian morphism $\mu$ and for some $\ell>0$.
    \item\label{2runs} $\rbwt(w)=2$.
    \end{enumerate}
\end{theorem}

 \begin{figure}
     \centering

$\begin{array}{ccccccc|c|}
\cline{8-8}
a & a & b & a & a & b & a & \textbf{b} \\
a & a & b & a & b & a & a & \textbf{b} \\
a & b & a & a & b & a & a & \textbf{b} \\
a & b & a & a & b & a & b & \textbf{a} \\
a & b & a & b & a & a & b & \textbf{a} \\
b & a & a & b & a & a & b & \textbf{a} \\
b & a & a & b & a & b & a & \textbf{a} \\
b & a & b & a & a & b & a & \textbf{a} \\
\cline{8-8}
\end{array}$

     \caption{BWT-matrix of the word $\fibmorph^4(a) = abaababa$: for each $i$, the $i$th row corresponds to the $i$th rotation of $\fibmorph^4(a)$ in lexicographic order, and the Burrows--Wheeler Transform $\bwt(\fibmorph^4(a)) = bbbaaaaa=b^3a^5$ is highlighted in bold in the last column. So, $r(abaababa)=2$.}
     \label{fig:BWT}
 \end{figure}


\section{New combinatorial properties of injective morphisms}\label{sec:new_comb_morphisms}
This section focuses on some combinatorial properties and characterizations of some classes of morphisms which are well-known in the context of coding theory and symbolic dynamics. The results provided in this section may be of independent interest and will later be related to BWT-run sensitivity in the next sections.

\subsection{Primitivity-preserving morphisms}

A morphism $\morph : \Sigma^* \rightarrow \Gamma^*$ is called \emph{primitivity-preserving} if for every $w\in Q(\Sigma^*)$, it holds that $\morph(w)\in Q(\Gamma^*)$, that is, primitive words are mapped to primitive words. Primitivity-preserving morphisms are injective, and the associated codes are known in the literature as \emph{pure codes}~\cite{Mitrana97}. Such codes have been introduced in~\cite{Restivo74} to study the relationships between locally testable languages and synchronizing properties of codes.

Given a morphism $\morph:\Sigma^*\rightarrow\Gamma^*$, we call a primitive word $w$ a \emph{$\morph$--power} if $\morph(w)=z^k$, for some primitive word $z$ and an integer $k>1$. Intuitively, it is a word that witnesses the non-primitivity-preserving property of a morphism. 
By $P^\morph$ we refer to the set of all  $\morph$--power words.
From the definition, hence, $P^\morph=\emptyset$ if and only if the morphism $\morph$ is primitivity-preserving. 

\begin{example}\label{ex:pd_power}
    Let $\pdmorph=(ab,aa)$ be the period-doubling morphism. The word $b$ is a $\pdmorph$--power, since $\pdmorph(b)=a^2$. Hence, $b\in P^\pdmorph$, and $\pdmorph$ is not primitivity-preserving.
\end{example}

\begin{example}
    Let $\morph=(a,bab)$. The word $ab$ is a $\morph$--power, since $\morph(ab)=(ab)^2$. Hence, $ab\in P^\morph$, and $\morph$ is not primitivity-preserving.
\end{example}

In this section, we prove a new characterization of the decompositions of binary primitivity-preserving morphisms.
To do so, we first recall the following lemma, characterizing the combinatorial structure of binary primitivity-preserving morphisms. 


\begin{lemma}[\cite{Huang}]
\label{le:pure=primitive}
   Let $\morph= (u,v)$ be an injective morphism, with $u,v$ two distinct primitive words. Then $\morph$ is a primitivity-preserving morphism if and only if all words in $\{u^nv^m \mid n,m\geq 1\}$ are primitive.
\end{lemma}



The following lemma describes what happens when the property of Lemma \ref{le:pure=primitive} is not verified. In particular, it considers the combinatorial structure of the non-primitive words generated by the morphism when applied to some primitive word distinct from a single letter. Recall that if $u^nv^m=z^k$, for some primitive word $z$ and $k>1$, then we can derive that $n=1$ or $m=1$ (see Theorem~\ref{lyndonS} in Appendix~\ref{ap:Lyndon}).
\begin{lemma}[\cite{RestR85,ShyrYu}]
\label{u*v*}
    Let $\morph = (u, v)$ be an injective morphism, and let $W = \{u^nv \mid n\geq1\} \cup \{uv^n \mid n>1\}$. Then, there is at most one primitive word $z$ and one integer $k>1$ such that $z^k\in W$, i.e. $|W\cap \overline{Q(\{a,b\}^*)}|\leq1$. Moreover, let $Y=\morph(Q(\{a,b\}^*)^{\geq 2})\cap \Pow{\{a,b\}^*}$. Then $Y=\conjclass(z^k)\cap \{u,v\}^+$.
\end{lemma}

The next lemma provides a characterization of the structure of the set $P^\morph$ for an injective morphism $\morph$. The proof can be found in Appendix \ref{ap:Lyndon}.

\begin{lemma}
\label{le:P_mu_structure}
    Let $\morph = (u, v)$ be an injective morphism, and let $W = \{u^nv \mid n\geq1\} \cup \{uv^n \mid n>1\}$. We can distinguish the two cases:
    \begin{enumerate}
        \item \label{Pmu_1}$|W\cap \Pow{\{a,b\}^*}|=0$. Then only one of the following occurs:
    \begin{enumerate}
        \item \label{Pmu_1a}$P^\morph = \emptyset$; 
        \item \label{Pmu_1b}$P^\morph = \{c\}$, for some $c\in\{a,b\}$;
        \item \label{Pmu_1c}$P^\morph = \{a,b\}$.
    \end{enumerate}
    \item \label{Pmu_2}$|W\cap \Pow{\{a,b\}^*}|=1$. Then there exists a unique $w\in\{a,b\}^*$ such that $\morph(w)\in W\cap \Pow{\{a,b\}^*}$, and only one of the following occurs:
    \begin{enumerate}
        \item \label{Pmu_2a}$P^\morph = \R{w}$;
        \item \label{Pmu_2b}$P^\morph = \R{w}\cup\{c\}$, for some $c\in \{a,b\}$.
    \end{enumerate}
    \end{enumerate}
\end{lemma}



Note that, among the cases described in Lemma \ref{le:P_mu_structure}, the Case~\ref{Pmu_1a} is the only one in which every primitivity-preserving morphism $\morph=(u,v)$ falls. In this case, both $u$ and $v$ are primitive words.  If the morphism $\morph=(u,v)$ is not primitivity-preserving and $W\cap \Pow{\{a,b\}^*}=\emptyset$, then it is easy to deduce from Lemma \ref{le:P_mu_structure} that either only one between $u$ and  $v$ is a non-primitive word (Case~\ref{Pmu_1b}), or both are non-primitive words (Case~\ref{Pmu_1c}).

A classification of the non-primitivity-preserving morphisms $\morph = (u, v)$ that fall in Cases~\ref{Pmu_2a} and~\ref{Pmu_2b}, with $W\cap \Pow{\{a,b\}^*}\neq \emptyset$, and their respective $\morph$-power words, can be derived from a result given in ~\cite[Theorem 8]{HolubRS23}. Such a classification is reported in the Appendix~\ref{ap:Lyndon} (Lemma \ref{le:Holub_classification}).

The set of primitivity-preserving morphisms is closed under composition, as shown in the following lemma.

\begin{lemma}
    Let $\morph_1:\Sigma^*\rightarrow\Gamma^*$, $\morph_2:\Gamma^*\rightarrow\Delta^*$ be two morphisms. If $\morph_1$ and $\morph_2$ are both primitivity-preserving, then $\morph_2\circ\morph_1$ is primitivity-preserving too.
\end{lemma}

However, it is possible to obtain primitivity-preserving morphisms even as a composition of morphisms that do not necessarily satisfy this property. The following proposition gives a complete characterization.

\begin{proposition}\label{prop:pp_morphism_decomposition}
    Let $\morph:\{a,b\}^*\rightarrow\{a,b\}^*$ be an injective morphism. The morphism $\morph$ is primitivity-preserving if and only if, for all $\psi,\chi:\{a,b\}^*\rightarrow\{a,b\}^*$ such that $\morph=\psi \circ \chi$, it holds that $\chi=(p,q)$ is a primitivity-preserving morphism and $\psi$ is an injective morphism such that $P^\psi\cap \{p,q\}^+ = \emptyset$.
\end{proposition}
\begin{proof}
    For the first direction, suppose by contraposition that either $\chi$ is not primitivity-preserving or $P^\psi\cap \{p,q\}^+ \neq \emptyset$.
    If $\chi$ is not primitivity-preserving, observe that there exists a primitive word $w\in\{a,b\}^*$ such that $\morph(w) = \psi(\chi(w)) = \psi(z^n) = \psi(z)^n$, for some primitive word $z$ and $n\geq2$.
    If $P^\psi\cap \{p,q\}^+ \neq \emptyset$, then there is at least one word $w\in \{a,b\}^*$ such that $\chi(w)\in P^\psi\cap \{p,q\}^+$, that is, $\morph(w) = \psi(\chi(w)) = z^n$ for some primitive word $z$ and some $n\geq2$.
    
    For the second direction, by hypothesis, we have that (i) a word $w$ is primitive if and only if $\chi(w)$ is primitive, and (ii) $\chi(w)\notin P^\psi$ for all $w\in\{a,b\}^*$. By combining these two assumptions, $w$ is primitive if and only if $\morph(w)=\psi(\chi(w))$ is primitive, and the thesis follows.
\end{proof}

\begin{example}
    Let $\morph=(abaa,aaab)$. It is easy to verify that $\morph=\pdmorph\circ\tau$, where $\pdmorph$ and $\tmmorph$ are the period-doubling morphism and the Thue--Morse morphism, respectively. 
    We have that $\pdmorph$ is not primitivity-preserving and  $P^{\pdmorph}=\{b\}$ (see Lemma~\ref{le:Holub_classification} in the Appendix~\ref{ap:Lyndon}).
    Since $\tmmorph$ is a primitivity-preserving morphism and $b\notin\{ab,ba\}^+$, from Proposition~\ref{prop:pp_morphism_decomposition} it follows that $\morph$ is primitivity-preserving too.
\end{example}

\begin{example}
Let $\psi = (aba,b)$, and consider the morphism $\morph= (abab, baba)=\psi\circ \tmmorph$, 
where $\tmmorph$ is the Thue--Morse morphism, which is primitivity-preserving. In this case, $\psi$ is not primitivity-preserving (since $\psi(ab)=(ab)^2$), nor is $\morph$. 
\end{example}

\subsection{Recognizable morphisms}

In this subsection, we focus on some structural and combinatorial properties of morphisms that generate words admitting a unique factorization in circular factors, similarly to the notion of circular code \cite{codesautomata}. 


Let $\La\subseteq\Sigma^*$. An injective morphism $\morph:\Sigma^*\rightarrow\Gamma^*$ is \emph{recognizable} on $\morph(\La)$ if for every non-empty word $w\in\morph(\La)$ and every word $w'\in\R{w}$, there exist, and are unique, $p\in\Gamma^+$, $q\in\Gamma^*$, $z\in\Sigma^*$, and $c\in\Sigma$, such that $w' = q\morph(z)p$ and $pq=\morph(c)$.
If $\La=\Sigma^*$, we simply say that $\morph$ is recognizable.


In other words, every image under a recognizable morphism $\morph$ has a unique circular factorization in words of $\morph(\La)$. Equivalently, a recognizable morphism on $\La$ can be regarded as an injective map on the necklaces over $\La$, i.e., for all $x,y \in \La$, it holds that $\R{\morph(x)} = \R{\morph(y)}$ if and only if $\R{x}=\R{y}$.

\begin{figure}[h!]
  \centering

  \begin{subfigure}[t]{0.3\textwidth}
    \centering
    \includegraphics[width=\linewidth]{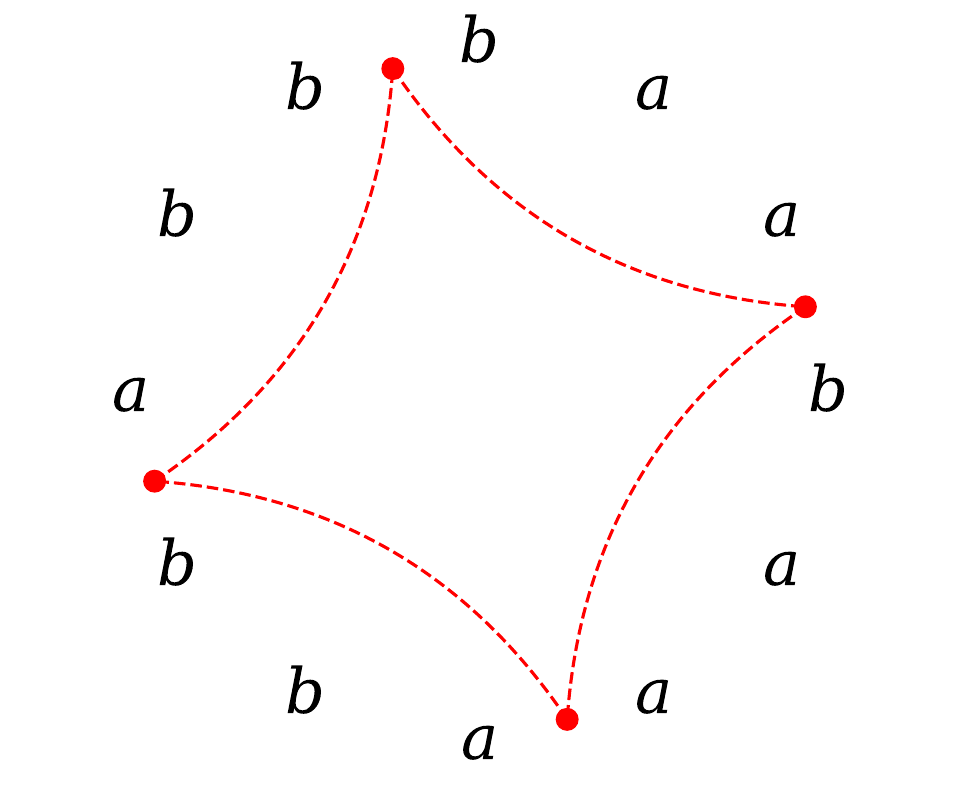}
    \caption{\centering $w=baaabbabbbaa$}
    \label{fig:baa_abb}
  \end{subfigure}
  \qquad
  \begin{subfigure}[t]{0.3\textwidth}
    \centering
    \includegraphics[width=\linewidth]{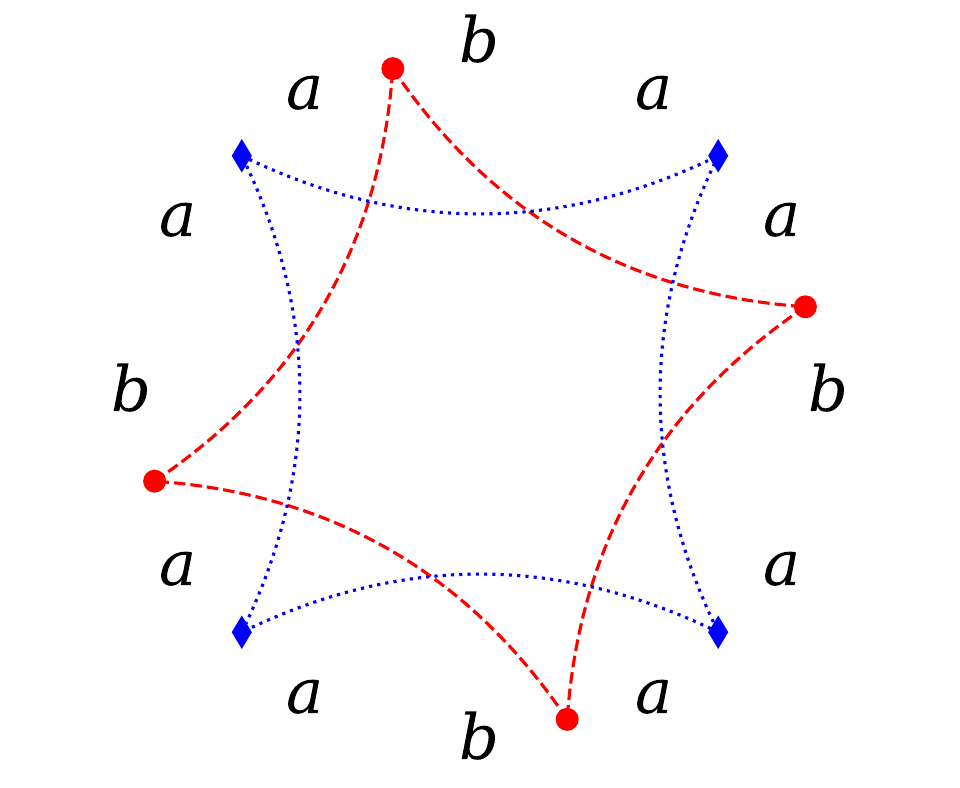}
    \caption{\centering $w=baabaabaabaa$}
    \label{fig:baa_aba}
  \end{subfigure}

  \caption{On the left, the unique circular factorization of $w=baaabbabbbaa$ into $\morph_1(a)=baa$ and $\morph_1(b)=abb$. On the right, two distinct circular factorizations of $w=baabaabaabaa$ into $\morph_2(a)=baa$ and $\morph_2(b)=aba$, respectively in blue and red.}
  \label{fig:recognizable_morphisms}
\end{figure}

\begin{example}
Let us consider the injective morphism $\morph_1=(baa, abb)$. Such a morphism is recognizable since every word in $\morph_1(\{a,b\}^*)$ has a unique circular factorization into the words $\morph_1(a)$ and $\morph_1(b)$, as shown in Figure \ref{fig:baa_abb}.
\end{example}

The recognizability of a morphism is well studied in the context of bi-infinite words and symbolic dynamics~\cite{BealPR_EJC24}. 
Here, it is adapted to necklaces, or circular words, which can be seen as periodic bi-infinite words. Note that most of the results known in the literature on bi-infinite words focus on the aperiodic case. Therefore, the results provided in this section can also be interpreted as contributions toward the less-explored setting of periodic bi-infinite words. 

The following lemma establishes the close relationship between recognizable morphisms on $\morph(\Sigma^*)$ and a property related to the so-called very pure codes, which are properly included in the class of pure codes \cite{Restivo74}. 

\begin{lemma}[{\cite[Proposition 7.1.1]{codesautomata}}]
\label{le:rec=circ}
An injective morphism $\morph:\Sigma^*\rightarrow\Gamma^*$ is recognizable if and only if for every $u,v\in \Gamma^*$, if $uv,vu\in \morph(\Sigma^*)$ then $u,v\in \morph(\Sigma^*)$.
\end{lemma}

In the case of binary injective morphisms, the next lemma provides a useful characterization of recognizable morphisms. The proof is given in Appendix \ref{ap:Lyndon}.

\begin{lemma}
\label{prop:purenotcircular=rotation}
    Every primitivity-preserving binary morphism $\morph=(u,v)$ is recognizable, unless $u$ and $v$ are conjugates.
\end{lemma}

\begin{example}
Consider the injective morphism $\morph_2=(baa, aba)$. Since $\morph_2(a)$ and $\morph_2(b)$ are conjugates, by Lemma \ref{prop:purenotcircular=rotation} $\morph_2$ is not recognizable on $\morph_2(\{a,b\}^*)$, as shown in Figure \ref{fig:baa_aba}, where two distinct circular factorizations of the word $baabaabaabaa$ are indicated. Indeed, $\morph_2(aaaa)$ and $\morph_2(bbbb)$ are equal up to rotations. Figure \ref{fig:blocchi-ab-ba} shows two distinct circular factorizations of $(ab)^6$ into $\tmmorph(a)$ and $\tmmorph(b)$. Hence, $\tmmorph$ is also not recognizable. One can note that $\morph_2=\Tilde{\fibmorph}\circ \tmmorph$, where $\Tilde{\fibmorph}=(ba,a)$, confirming the characterization established in the following theorem.
\end{example}

The following result shows an important structural property of the primitivity-preserving morphisms that are not recognizable. In particular, we prove that they can always be obtained by composing another morphism with the Thue--Morse morphism.

\begin{theorem}
\label{th:primitivitynoncircular=thuemorse}
    Let $\morph:\{a,b\}^*\rightarrow\Gamma^*$ be a primitivity-preserving binary morphism. Then exactly one of the following cases occurs:
    \begin{enumerate}
        \item $\morph$ is recognizable;
        \item $\morph = \psi \circ \tau$ for some injective morphism $\psi$,  where $\tmmorph$ is the Thue--Morse morphism.
    \end{enumerate}
\end{theorem}

\begin{proof}
We show that, under the hypothesis of the theorem, it holds that $\morph=(uv,vu)$ if and only if $\morph = \psi\circ \tmmorph$ for some injective morphism $\psi$, and by Lemma~\ref{prop:purenotcircular=rotation}, the thesis directly follows. 
For the first direction, one can define $\psi=(u,v)$, and therefore $\morph=(\psi(\tmmorph(a)),\psi(\tmmorph(b))) = (\psi(ab),\psi(ba)) = (uv, vu)$.
For the other direction, if $\morph=\psi\circ\tau$, then $\morph = (\psi(\tau(a)), \psi(\tau(b))) = (\psi(ab),\psi(ba)) = (\psi(a)\psi(b), \psi(b)\psi(a))$,  and the thesis follows. 
\end{proof}

\subsection{Synchronizing morphisms}

An important notion in the context of injective morphisms is that of synchronization pair, which intuitively marks a position within a factor of a morphic image where the boundary between two codewords can be uniquely identified. Synchronization provides a way to ``align'' a segment of the morphic image of a circular word with the images of the letters of the alphabet. 

    Let $\morph:\Sigma^*\rightarrow\Gamma^*$, $\La\subseteq \Sigma^*$, and $u\in \cfact(\morph(\La))$.
    We say that $(u_1,u_2)$ is a \emph{synchronization pair} of $u$ on $\morph(\La)$ if $u=u_1u_2$ and, for all $v_1,v_2\in \Gamma^*$ and $f\in \cfact(\La)$, $v_1uv_2=\morph(f)$ implies $v_1u_1=\morph(f_1)$ and $u_2v_2 = \morph(f_2)$, for some $f_1,f_2\in \cfact(\La)$ such that $f_1f_2=f$.

Observe that a morphism $\morph$ is recognizable on $\morph(\La)\subseteq\morph(\Sigma^*)$ if and only if every word $w\in\morph(\La)$ admits at least one synchronization pair, since from it one can uniquely recover the preimage $w'=\morph^{-1}(w)$, up to rotations.

The following notion of synchronization with delay gives a quantitative measure of the width of a window sliding along the morphic image of a circular word that guarantees the detection of a synchronization point.  

    We say that a morphism $\morph:\Sigma^*\rightarrow\Gamma^*$ is \emph{synchronizing with delay $k>0$} for $w\in\morph(\Sigma^*)$ if every circular factor of $w$ of length at least $k$ admits a synchronization pair. Given $\La\subseteq\Sigma^*$, we say that $\morph$ is \emph{synchronizing with delay $K>0$} for $\morph(\La)$ if it is synchronizing with a finite delay for every $w\in \La$ and
    $$\sup\left\{ \min_{x\in \La}\{k \mid \morph \text{ is synchronizing with delay $k$ for $\morph(x)$} \}\right\}\leq K.$$

It has been proved \cite[Theorem 5.1]{Restivo74} that a morphism is recognizable if and only if it is synchronizing with finite delay for $\morph(\Sigma^*)$. The following example shows that the recognizability of a morphism on a proper subset of $\morph(\Sigma^*)$ does not necessarily imply being synchronizing with finite delay for that subset. 

\begin{example}
    \label{ex:thuemorsenonrecognizable}
    Let $\tmmorph$ be the Thue--Morse morphism.
 Figure \ref{fig:blocchi-ab-ba} shows that the Thue--Morse morphism $\tmmorph$ is not recognizable. In fact, $\tmmorph(aaaaaa)$ and $\tmmorph(bbbbbb)$ are equal up to rotations and produce distinct circular factorizations. 
However, $\tmmorph$ is recognizable on $\tmmorph(\La)$, where $\La=\{a^n b,b^na\mid n>0\}$, as shown in Figure \ref{fig:TM_recognizable}. Observe that even though $\tmmorph$ is not recognizable, $\tmmorph$ is recognizable on $\tmmorph(\La)$, since $(\tmmorph(a),\tmmorph(b))$ and $(\tmmorph(b),\tmmorph(a))$ are synchronization pairs that occur in every word of $\tmmorph(\La)$. In the figure, the unique circular factorization of $\tmmorph(a^5b)=(ab)^5ba$ is depicted; the two black squares identify the synchronization pairs $(b, b)$ and $(a, a)$. Moreover, $\tmmorph$ is synchronizing with delay $11$ on $(ab)^5ba$; more in general it is synchronizing with delay $2n+1$ for $(ab)^nba$ and for $(ba)^nab$. However, $\tmmorph$ is not synchronizing with finite delay for $\tmmorph(\La)$, since the supremum of all minimum finite delays for all the words in $\La$ is unbounded. 

Let $\La'=\tmmorph(\{a,b\}^*) = \{ab,ba\}^*$. Unlike the previous cases, the synchronization pairs $(a,a)$ and $(b,b)$ occur in every word of $\cfact(\tmmorph(\La'))$ of length at least $5$, hence $\tmmorph$ is synchronizing with delay $5$ for $\tmmorph(\La')$. In fact, as shown in Figure \ref{fig:TM_finite_delay}, every factor of $\tmmorph(abbaab)=abbabaababba$ having length at least $5$ contains a black square that identifies a synchronization pair.
\end{example}

\begin{figure}[h!]
  \centering
    
  \begin{subfigure}[t]{0.3\textwidth}
    \centering
    \includegraphics[width=\linewidth]{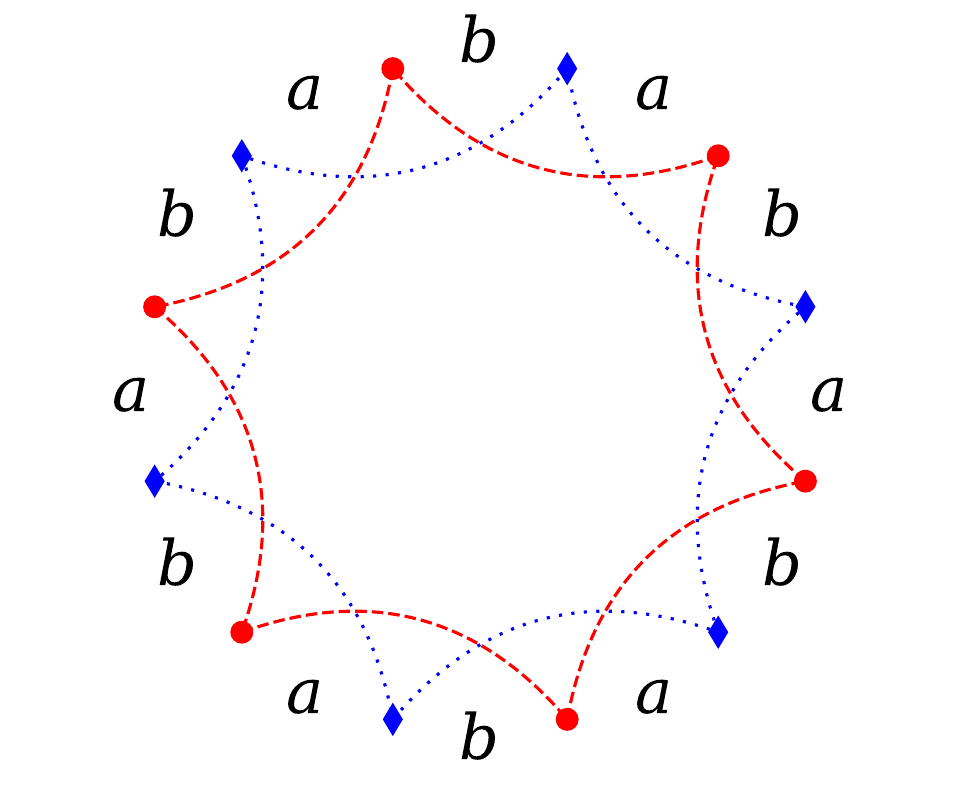}
    \caption{\centering $w=abababababab$}
    \label{fig:blocchi-ab-ba}
  \end{subfigure}
  \hfill
  \begin{subfigure}[t]{0.3\textwidth}
    \centering
    \includegraphics[width=\linewidth]{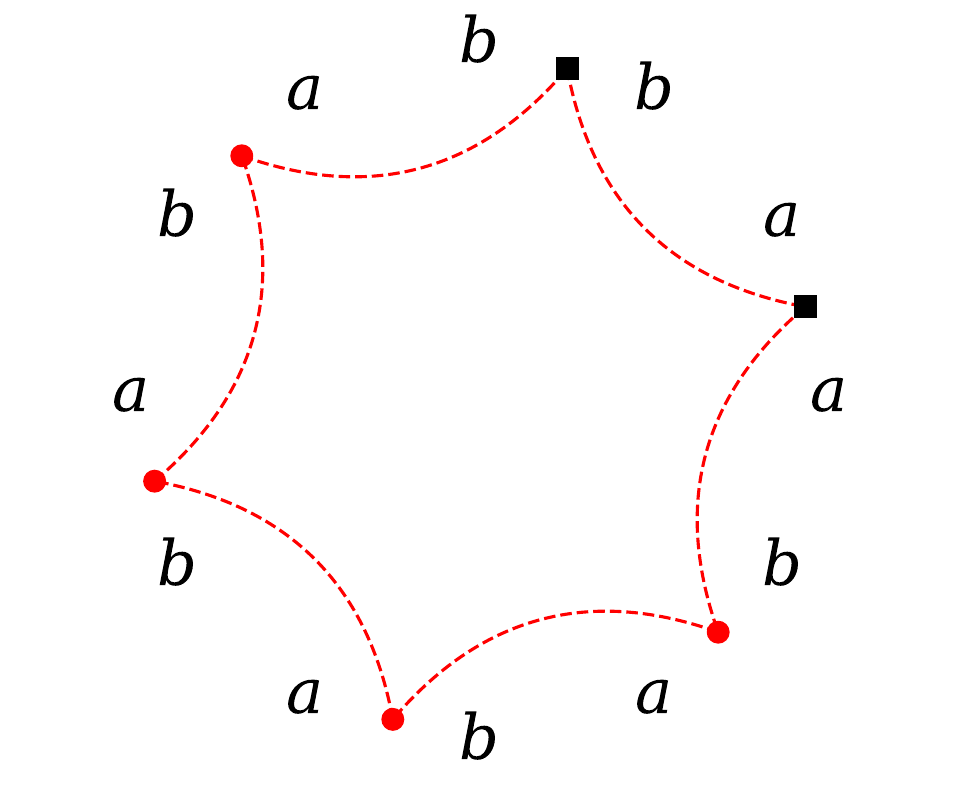}
    \caption{\centering $w=abababababba$}
    \label{fig:TM_recognizable}
  \end{subfigure}
  \hfill
  \begin{subfigure}[t]{0.3\textwidth}
    \centering
    \includegraphics[width=\linewidth]{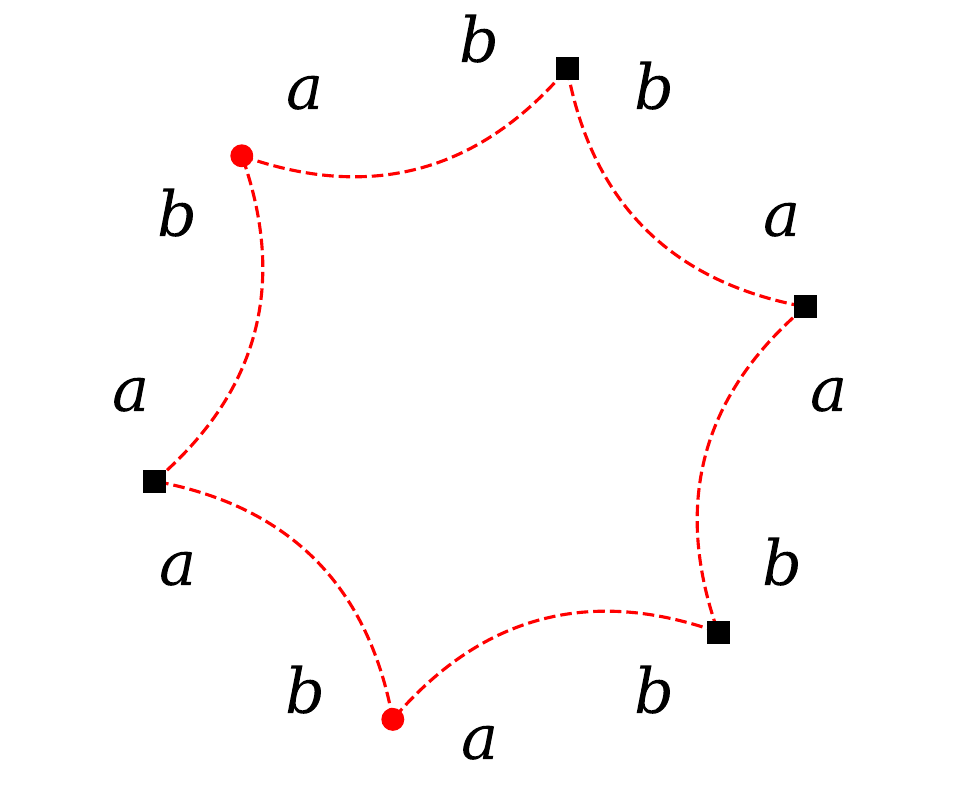}
    \caption{\centering$w=abbabaababba$}
    \label{fig:TM_finite_delay}
  \end{subfigure}

  \caption{Circular factorizations into $\tmmorph(a)=ab$ and $\tmmorph(b)=ba$ are depicted, where $\tmmorph$ is the Thue--Morse morphism. On the left, two distinct circular factorizations of $(ab)^6$ in blue and red, respectively; in the center, the unique circular factorizations of $w=abababababba$; on the right, the unique circular factorizations of $w=abbabaababba$. Each black square identifies a synchronization pair.}
  \label{fig:sync_morphisms}
\end{figure}

We now give a new combinatorial characterization of synchronizing morphisms with finite delay on $\morph(\La)$, for any $\La\subseteq \{a,b\}^*$.
This characterization is based on the powers of single letters occurring in $\La$ in the case of non-recognizable primitivity-preserving morphisms, while it is based on the $\morph$-power words in the case of non-primitivity-preserving morphisms.



\begin{lemma}
\label{le:whenNonRecAreSynch}
    Let $\morph:\{a,b\}^*\rightarrow\Gamma^*$ be a non-recognizable primitivity-preserving morphism and let $\cfact_a=\cfact(\La)\cap\{a\}^*$ and $\cfact_b=\cfact(\La)\cap\{b\}^*$, where $\La\subseteq\{a,b\}^*$. 
    Then $\morph$ is synchronizing with finite delay on $\morph(\La)$ if and only if at least one of the sets $\cfact_a$ or $\cfact_b$ is finite.
\end{lemma}
\begin{proof}
    By Lemma~\ref{prop:purenotcircular=rotation}, there exist $p,q\in\Gamma^*$, with $p\neq q$, such that $\morph=(pq,qp)$. Since $\morph$ is injective, $pq\neq qp$; hence
    all words in $\R{pq}$ are primitive, and $(pq,qp)$ and $(qp,pq)$ are synchronization pairs.
    For the first direction, observe that if both sets $\cfact_a$ and $\cfact_b$ are infinite, then all the factors in $\cfact(\morph(\La))\cap(\{pq\}^*\cup\{qp\}^*)$ have no synchronization pairs.
    For the other direction, let us assume that $\cfact_a$ is finite (the case $\cfact_b$ finite can be treated analogously).
    Let $t=\max\{i\geq0\mid a^i\in\cfact_a\}$.
    Then, if $(pq)^{t+1}\in\cfact(\morph(\La))$, the pair $((pq)^tp,q)$ is synchronizing of $(pq)^{t+1}$ on $\morph(\La)$, since $(pq)^{t+1} = p\morph(b^t)q = \morph(a^{t+1})$ but $a^{t+1}\notin\cfact(\La)$. 
    Finally, since in every factor in $\cfact(\morph(\La))$ of length $k = |\morph(a^{t+2})|$ there is an occurrence of one of the factors with a synchronization pair listed above, the thesis follows.
\end{proof}


By using analogous techniques, one can prove that for any non-primitivity-preserving morphism, there exists a $k>0$ such that each $k$-length factor $w$ in $\cfact(\morph(\Sigma^*))$ has a synchronization pair, unless $w\in\cfact(\morph(z^*))$ for some $z\in P^\morph$.
The structure of the non-primitivity-preserving morphisms, detailed in Lemma~\ref{le:Holub_classification} (see Appendix~\ref{ap:Lyndon}), is used. 

\begin{lemma}
\label{le:non_prim_sync_pair}
    Let  $\morph:\{a,b\}^*\rightarrow\Gamma^*$ be a non-primitivity-preserving morphism. Then, there exists an integer $k>0$ such that every factor $w\in\Gamma^k\cap(\cfact(\morph(\Sigma^*)) \setminus \cfact(\{\morph(z^*)\mid z\in P^\morph\}))$ has a synchronization pair.
\end{lemma}

From the previous lemma, the following result can be derived:

\begin{theorem}
\label{le:whenNonPrimPresAreSynch}
    Let  $\morph:\{a,b\}^*\rightarrow\Gamma^*$ be a non-primitivity-preserving morphism and let $\La\subseteq\{a,b\}^*$.
    Then $\morph$  is synchronizing with finite delay on $\morph(\La)$ if and only if the set $\cfact(\La)\cap \{w^*\mid w\in P^\morph\}$ is finite.
\end{theorem}

\section{Sensitivity of the measure \texorpdfstring{$r$}{r} to the application of morphisms}\label{sec:sensitivity}

Let $\mu$ be a morphism and $w$ a word.
In~\cite{Fici23} we defined:  $$\Delta^+_\morph(w)=\rbwt(\morph(w)) -\rbwt(w)$$ and  $$\Delta^{\times}_\morph(w)=\frac{\rbwt(\morph(w))}{\rbwt(w)}.$$

Notice that $\Delta^+_\morph(w)$ may be negative for some word $w$. For example, let $\morph$ be the morphism over a 3-letter alphabet $\{a,b,c\}$ defined as $\morph= (b,a,c)$ and let $w=bcba$. One has that $r(w)=|\rle(\bwt(bcba))|=|\rle(bcab)|=4$ and $r(\morph(w))=r(acab)=|\rle(\bwt(acab))|=|\rle(cbaa)|=3$. However, when $\morph$ is defined over a binary alphabet, one can prove that $\Delta^+_\morph(w)$ is always non-negative \cite[Theorem 14]{Fici23}.



\begin{definition}The {\em BWT additive sensitivity function} and {\em BWT multiplicative sensitivity function} for a morphism $\mu$ are,  respectively, the functions $$AS_\mu(n) = \max_{w \in \Sigma^n}(\Delta_{\mu}^+(w))\text{ \ and \ } MS_\mu(n) = \max_{w \in \Sigma^n}(\Delta_{\mu}^\times(w))$$
\end{definition}

Note that the additive sensitivity function is always a non-negative function, as for every $n$, $\Delta_{\mu}^+(a^n)\geq 0$ for any letter $a$.

\begin{example}
Let us consider the period-doubling morphism $\pdmorph$. Let us compute the value of the BWT additive sensitivity function for $\pdmorph$ when $n=5$. From Table~\ref{table:pd5}, it is possible to conclude that $AS_\pdmorph(5)=MS_\pdmorph(5)=2$.
\end{example}
\begin{table} 
\centering
\begin{tabular}{|c|c|c|c|c|c|}
\hline
$w$ & $\bwt(w)$ & $r(w)$ & $\pdmorph(w)$ & $\bwt(\pdmorph(w))$ & $r(\pdmorph(w))$\\
\hline
$aaaab$ & $baaaa$ & $2$ & $ababababaa$ & $babbbaaaaa$ & $4$\\
$aaabb$ & $baaba$ & $4$ & $abababaaaa$ & $baaabbaaaa$ & $4$\\
$aabab$ & $bbaaa$ & $2$ & $ababaaabaa$ & $bbaabaaaaa$ & $4$\\
$aabbb$ & $babba$ & $4$ & $ababaaaaaa$ & $baaaaabaaa$ & $4$\\
$ababb$ & $bbbaa$ & $2$ & $abaaabaaaa$ & $babaaaaaaa$ & $4$ \\
$abbbb$ & $bbbba$ & $2$ & $abaaaaaaaa$ & $baaaaaaaaa$ & $2$\\
\hline
\end{tabular}
\caption{The first column contains the list of all words of length $5$, up to rotations. This is not restrictive, since rotations of the same word have the same value of $r$. The columns $r(w)$ and $r(\pi(w))$ are used to compute $AS_{\pi}(5)$.}\label{table:pd5}
\end{table}

The following lemma shows that cyclic morphisms produce words with a fixed number of BWT-runs, whatever the words on which they are applied. 

\begin{lemma}\label{le:cyclic_sensitivity}
Let $\morph:\{a,b\}^*\rightarrow\Gamma^*$ be a cyclic morphism. Then, there exist two constants $k^+_{\morph}, k^\times_\morph$, which depend on $\morph$, such that $\bwas(n) = k^+_\morph$ and $\bwms(n) = k^\times_\morph$, for all $n\geq2$.
\end{lemma}
\begin{proof}
Recall that a binary morphism is cyclic if and only if there exist two integers $t_1,t_2>0$ and a non-empty word $z\in\Gamma^+$ such that $\morph(a)=z^{t_1}$ and $\morph(b)=z^{t_2}$. Hence, for each word $w\in\Sigma^+$ it holds that $r(\morph(w)) = r(z^{|w|_{a}t_1+|w|_{b}t_2}) = r(z)$.
Let us fix the claimed constants $k_\morph = r(z)-2$ and $k'_\morph=r(z)/2$. 
For all $n\geq 2$, let us consider the word $s_n=a^{n-1}b$. By Lemma~\ref{bwsturmian}, it follows that $r(s_n)=2$.
The proof follows by observing that since $r(\morph(w))$ is constant, the values of $\Delta_\morph^+$ and $\Delta_\morph^\times$ are maximal when $r(w)$ assumes the smallest value, that in the case of binary words is $2$, i.e., $AS_\morph(n) = \max_{w \in \Sigma^n}(r(\morph(w) - r(w)) = r(z)-r(s_n) = k^+_\morph$ and $MS_\morph(n) = \max_{w \in \Sigma^n}(r(\morph(w)/r(w)) = r(z)/r(s_n) =k^\times_\morph$.
\end{proof}

\begin{example}
Let us consider the cyclic morphism $\mu=(ababbba,(ababbba)^2)$. It is possible to verify that for every $w\in\{a,b\}^+$, one has $\mu(w)=(ababbba)^p$, for some integer $p>0$ depending on $w$. This means that $r(\mu(w))=r(ababbba)=6$ for every $w\in\{a,b\}^+$. For every length $n$, we can consider the word $a^{n-1}b$. We have $r(a^{n-1}b)=2$, which is the lowest value that $r$ can take on a binary word. Then, $AS_{\mu}(n)=6-2=4$ and $MS_{\mu}(n)=6/2=3$, for $n\geq 2$. 
\end{example}

The following characterization of Sturmian morphisms in terms of the BWT additive sensitivity function was proved in~\cite{Fici23}.

\begin{proposition}[\cite{Fici23}]
\label{prop:sturmian0bwt}
Let $\morph$ be a binary injective morphism. Then $AS_\morph(n)=0$ for every $n\geq 2$ if and only if $\morph$ is a Sturmian morphism.
\end{proposition}

In the same paper, we showed that the Thue--Morse morphism $\tmmorph$ increases by 2 the BWT-runs of every binary word, while in the case of the period-doubling morphism $\pdmorph$, for each $n \geq 2$ we can find an $n$-length word $w$ for which $\Delta_\pdmorph^+(w) = \Theta(\sqrt{n})$.
We summarize these results in the following proposition.

\begin{proposition}[\cite{Fici23}]
\label{prop:tmpd}
Let $\tmmorph$ and $\pdmorph$ be the Thue--Morse and the period-doubling morphisms, respectively. The following properties hold:
\begin{enumerate}
\item \label{item:tm} $AS_\tmmorph(n)=2$, for all $n\geq 2$;
\item \label{item:pd} $AS_\pdmorph(n)=\Omega(\sqrt{n})$.
\end{enumerate}
\end{proposition}

Note that $\tmmorph$ is not the only morphism for which the additive sensitivity function is $2$. In~\cite{Fici23} it is proved that this property also holds for the \emph{Thue--Morse-like} morphisms $\tau_{p,q} = (ab^p, ba^q)$, for some $p,q > 0$, and any composition of these morphisms with any Sturmian morphism.  

\begin{example}
Let us consider the morphism $\morph=(abbaab,ababba)$. It is possible to verify that $\morph=\tmmorph \circ \fibmorph \circ \tmmorph$, where $\tmmorph$ and $\fibmorph$ are, respectively, the Thue--Morse and the Fibonacci morphism. By using Propositions~\ref{prop:sturmian0bwt} and~\ref{prop:tmpd}, item~\ref{item:tm}, it follows that $AS_\morph(n)=4$ for all $n\geq 2$.
\end{example}

\begin{figure}[ht]
    \centering
\[
\begin{array}{cccc|l|}
\cline{5-5}
a&a&b&a & b\\
a&b&a&a & b\\
a&b&a&b & a\\
b&a&a&b & a\\
b&a&b&a & a\\
\cline{5-5}
 \end{array}
\qquad\qquad
 \begin{array}{cccccccccccccc|l|}
    \cline{15-15}
\textbf{a}&\textbf{a.}&\textbf{a} & b&a.&b&a&a.&a&b&a.&b&a&a.& \textbf{b}\\
\textbf{a}&\textbf{a.}&\textbf{a}&b&a.&b&a&a.&b&a&a.&a&b&a. & \textbf{b}\\
\hdashline
\textbf{a} &\textbf{a.} &\textbf{b}&\textbf{a}&\textbf{a.}&\textbf{a}&b&a.&b&a&a.&a&b&a. & \textbf{b}\\
\hdashline
\rowcolor{gray!50}\textbf{a.}&\textbf{a}&\textbf{b}&\textbf{a.}&\textbf{b}&a&a.&a&b&a.&b&a&a.&b & \textbf{a}\\
\rowcolor{gray!50}\textbf{a.}&\textbf{a}&\textbf{b}&\textbf{a.}&\textbf{b}&a&a.&b&a&a.&a&b&a.&b & \textbf{a}\\
\hdashline
\rowcolor{gray!50}\textbf{a.}&\textbf{b}&\textbf{a}&\textbf{a.}&\textbf{a}&b&a.&b&a&a.&a&b&a.&b & \textbf{a}\\
\rowcolor{gray!50}\textbf{a.}&\textbf{b}&\textbf{a}&\textbf{a.}&\textbf{a}&b&a.&b&a&a.&b&a&a.&a & \textbf{b}\\
\hdashline
\rowcolor{gray!50}\textbf{a.}&\textbf{b}&\textbf{a}&\textbf{a.}&\textbf{b}&\textbf{a}&\textbf{a.}&\textbf{a}&b&a.&b&a&a.&a & \textbf{b}\\
\hdashline
\rowcolor{gray!20}\textbf{a}&\textbf{b}&\textbf{a.}&\textbf{b}&a&a.&a&b&a.&b&a&a.&b&a & \textbf{a.}\\
\rowcolor{gray!20}\textbf{a}&\textbf{b}&\textbf{a.}&\textbf{b}&a&a.&b&a&a.&a&b&a.&b&a & \textbf{a.}\\
\hdashline
\rowcolor{gray!20}\textbf{b}&\textbf{a}&\textbf{a.}&\textbf{a}&b&a.&b&a&a.&a&b&a.&b&a & \textbf{a.}\\
\rowcolor{gray!20}\textbf{b}&\textbf{a}&\textbf{a.}&\textbf{a}&b&a.&b&a&a.&b&a&a.&a&b & \textbf{a.}\\
\hdashline
\rowcolor{gray!20}\textbf{b}&\textbf{a}&\textbf{a.}&\textbf{b}&\textbf{a}&\textbf{a.}&\textbf{a}&b&a.&b&a&a.&a&b & \textbf{a.}\\
\hdashline
\textbf{b}&\textbf{a.}&\textbf{b}&a&a.&a&b&a.&b&a&a.&b&a&a. & \textbf{a}\\
\textbf{b}&\textbf{a.}&\textbf{b}&a&a.&b&a&a.&a&b&a.&b&a&a. & \textbf{a}\\
    \cline{15-15}
 \end{array}
 \]

    \caption{Comparison of the BWT--matrices for the word $w = aabab$ (on the left) and its image after application of the morphism $\morph=(baa,aba)$ (on the right).  The dashed lines partition the rotations according to the shortest prefixes with at least one synchronization pair (highlighted in bold). The rotations in light gray correspond to the words in $\morph(\R{w})$. The rotations in dark gray correspond to the rotations where $\bwt(w)$ is spelled in reverse order.}
    \label{BWT}
\end{figure}

\section{Characterization of binary BWT-run preserving morphisms}\label{sec:characterization_morphisms}

As a main result of this paper, we characterize the binary morphisms having a bounded BWT additive sensitivity function. In particular, we prove that they coincide with the primitivity-preserving morphisms. 

\begin{definition}
Let $k\geq 0$ be an integer. An acyclic morphism $\morph:\Sigma^*\rightarrow\Gamma^*$ is called \emph{$k$-BWT-run preserving} if for all $n\geq |\Sigma|$, $\bwas(n)\leq k$. We simply say \emph{BWT-run preserving} if such a $k$ exists.
\end{definition}



We first give a lemma, in which we prove that the finite-delay synchronization of a morphism on the images of a language results in a bounded increase in the number of BWT-runs.  The proof can be found in the Appendix \ref{ap:l32}. 

\begin{lemma}
\label{le:recmorphdelay->BWTrunbounded}
    Let $\La\subseteq\Sigma^*$, where $\Sigma =\{a,b\}$, and let $\morph:\Sigma^* \rightarrow\Gamma^*$ be synchronizing with delay $k>0$ on $\morph(\La)$.
    Then, there exists $k'>0$ such that $\Delta^+_{\morph}(u)\leq k'$ for all $u\in \La$.
\end{lemma}

The following lemma proves one direction of the main result.
\begin{lemma}\label{lem:primitivitypreserving->BWTrunpreserving}
    Let $\morph:\{a,b\}^*\rightarrow\Gamma^*$ be an injective morphism.
    If $\morph$ is primitivity-preserving, then $\morph$ is BWT-run preserving.
\end{lemma}
\begin{proof}
    If $\morph$ is primitivity-preserving, then by Theorem~\ref{th:primitivitynoncircular=thuemorse}, either (i) $\morph$ is recognizable or (ii) there exist an integer $t>0$ and a morphism $\psi:\{a,b\}^*\rightarrow\Gamma^*$ such that $\morph = \psi \circ \tmmorph^t$ and $\psi\neq\psi'\circ\tmmorph$ for all $\psi':\{a,b\}^*\rightarrow\Gamma^*$. 
     If we fall in case (i), the thesis follows from the equivalence between recognizable morphisms and synchronizing morphisms with bounded delay~\cite[Theorem~5.1]{Restivo74} and Lemma~\ref{le:recmorphdelay->BWTrunbounded}.
     
    If instead we fall in case (ii), then by Proposition~\ref{prop:tmpd}, it follows that $\tmmorph$ increases the BWT runs by (at most) 2. Hence, the thesis is equivalent to showing that there exists $k\geq0$ such that $r(\psi(w))\leq r(w)+k$, for every $w\in\tmmorph^t(\{a,b\}^*)$. This would prove that the BWT additive sensitive function is bounded by $k+2t$. 
    By Proposition~\ref{prop:pp_morphism_decomposition}, we can distinguish between two subcases: (ii.a) $\psi$ is recognizable and (ii.b) $\psi$ is not primitivity-preserving and $\tmmorph^t(a)\notin P^\psi$.
   If (ii.a), the proof follows analogously to (i).
    If (ii.b), then observe that  $\cfact(\tmmorph^t(\Sigma^*))$ contains a finite number of powers of elements from $P^\psi$, and the proof follows by Theorem~\ref{le:whenNonPrimPresAreSynch}.
\end{proof}



Now we prove the opposite direction. We consider a class of morphisms that we use to decompose a generic morphism. For any $p > 1$,  let $\rho_{p}:\{a,b\}^*\rightarrow\{a,b\}^*$ denote the injective morphism $(a, b^p)$.
Observe that if $p>1$, then $\rho_p$ is not primitivity-preserving.

In the following proposition, we prove that such morphisms have an unbounded additive sensitivity function. The proof is given in the Appendix~\ref{ap:prop3435}. 






\begin{proposition}
    Let $\rho_p = (a, b^p)$, for some $p>1$. Then, $AS_{\rho_p}(n)=\Omega(\sqrt{n})$.\label{prop:rho_unbounded_AS}
\end{proposition}

In the following proposition, we consider a larger class of morphisms with an unbounded additive sensitivity function. The proof can be found in the Appendix~\ref{ap:prop3435}. 

\begin{proposition}\label{prop:unbounded_increase_wk_power_morphisms}
Given an injective morphism $\morph:\{a,b\}^*\rightarrow\Gamma^*$, let $u,v\in \Prim{\Gamma^*}$ and $p,q\geq1$ such that $\morph = (u^p,v^q)$. Then, 
$$\mu = \eta \circ \rho_q \circ E \circ \rho_p \circ E$$
where $\eta = (u, v)$.
Moreover, if $pq > 1$, 
then $AS_{\mu}(n) = \Omega(\sqrt{n})$.
\end{proposition}

The following lemma shows that if a morphism has bounded additive sensitivity, then it is primitivity-preserving. 
\begin{lemma}
\label{lem:unboundednonprimitivepreserving}
    Let $\morph:\{a,b\}^*\rightarrow\Gamma^*$  be an injective non-primitivity-preserving binary morphism.
    For each $k>0$, there exists a word $w$ such that $\Delta_\morph^{+}(w) > k$.   
\end{lemma}
\begin{proof}
    If $\morph(a)$ or $\morph(b)$ are not primitive, then the thesis follows by Proposition~\ref{prop:unbounded_increase_wk_power_morphisms}, so let us assume that $\morph(a),\morph(b)\in\Prim{\Gamma^*}$.  

    Recall that a Lyndon word is a primitive word that is lexicographically smaller than all its proper conjugates.
    Since $\morph$ is injective, not primitivity-preserving, and both images are primitive words, there exists some Lyndon word $x \in P^\morph$ such that $|x| > 1$, $\morph(x) = z^t$ for some $t > 1$, and $z$ is primitive. Let $\psi = (a,x)$ and $\eta = (\mu(a), z^t)$. Observe that: i) $\morph(\psi(a)) = \morph(a)$; and ii) $\morph(\psi(b)) = \morph(x) = z^t$. Hence, $\eta=\morph \circ \psi$. 
    Then, by Proposition~\ref{prop:unbounded_increase_wk_power_morphisms}, $\eta$ there exists a word $w\in\{a,b\}^n$ such that $\Delta_\eta^+(w)= \Theta(\sqrt{n})$.
    Since the concatenation $uv$ of two Lyndon words $u$ and $v$, with $u<v$, is a Lyndon word (see~\cite{lothaire1997combinatorics}), then, for every $m\geq1$, $a^mx$ and $ax^m$ are Lyndon words, hence, by Lemma~\ref{le:pure=primitive}, the morphism $\psi=(a,x)$ is primitivity-preserving, and by Lemma~\ref{lem:primitivitypreserving->BWTrunpreserving} $\psi$ is BWT-run preserving. 
    Finally, one has
    $\Delta_\eta^+(w) = \Delta_\morph^+(\psi(w)) + \Delta_\psi^+(w) = \Delta_\morph^+(\psi(w)) + O(1) = \Theta(\sqrt{n})$, and the thesis follows.
\end{proof}


\begin{example}
Let $\morph = (ba,ababaa)$. Let $x = aab$ and $z = babaa$, where $x$ is Lyndon and $z$ is primitive. It holds that $$\morph (x) = \morph(aab) = ba\cdot ba \cdot ababaa = (babaa)^2 = z^2.$$ 
We define the morphisms $\psi = (a, aab)$ and $\eta = (ba,(babaa)^2)$, as described in the proof of Lemma~\ref{lem:unboundednonprimitivepreserving}.
Indeed, $\eta = \morph \circ \psi$, as $\morph(\psi(a)) = \morph(a) = ba = \eta(a)$ and $\morph(\psi(b)) = \morph(aab) =(babaa)^2 = \eta(b)$. The morphism $\eta$ can be written as $\eta = (ba, babaa) \circ \rho_2$, 
and by Proposition~\ref{prop:unbounded_increase_wk_power_morphisms} there exists $w\in\{a,b\}^*$ such that $r(\eta(w)) - r(w) = \Theta(\sqrt{n})$. 
On the other hand, $\psi$ is primitivity-preserving, so it must be the case that $\rbwt(\morph(\psi(w)))-\rbwt(\psi(w)) = \Theta(\sqrt{n})$.
\end{example}

From Lemmas~\ref{lem:primitivitypreserving->BWTrunpreserving}  and~\ref{lem:unboundednonprimitivepreserving}, the main result of the paper can be derived.

\begin{theorem}
    Let $\morph:\{a,b\}^*\rightarrow\Gamma^*$ be an injective morphism.  Then $\morph$ is BWT-run preserving if and only if it is primitivity-preserving.
\end{theorem}
    
Finally, we can show that there exists a finite test case, as stated in the following theorem.

\begin{theorem}
    Let $\morph:\{a,b\}^*\rightarrow\Gamma^*$ be an acyclic morphism. 
    It is decidable in polynomial time in the size of $\morph$ whether $\morph$ is BWT-run preserving.
\end{theorem}
\begin{proof}
    By Lemma~\ref{le:pure=primitive}, to decide whether a given morphism $\morph=(u,v)$, for some $u,v\in\Gamma^+$, is primitivity-preserving, we have to check the primitiveness of all the possible non-trivial solutions of the equation $u^\ell v^m = z^n$.
    Let $t_{\max} = \max\{|u|,|v|\}$ and $t_{\min}=\min\{|u|,|v|\}$. Then, by Theorem~\ref{th:holub} (see Appendix~\ref{ap:Lyndon}), there are at most $O(t_{\max}/t_{\min})$ words to check, each of these having length $\Theta(|u|+|v|)$. Since the primitiveness can be checked in linear time in the size of the words, the total time complexity is $O(t_{\max}^2/t_{\min})$.
\end{proof}

\section{Morphisms with bounded multiplicative sensitivity}\label{sec:multiplicative_sensitivity}

Even though in the case of binary morphisms the additive sensitivity is not always bounded by a constant, it is natural to wonder whether the multiplicative sensitivity is.
As shown in the following example, this is not the case when the alphabet size is greater than 2.

\begin{example}\label{ex:big_alphabet_ms_unbounded}
Let $f_k^\$ = \varphi^{k}(a)\$$ be the $k$-th Fibonacci word with a letter $\$$ such that $\$ < a < b $ appended. Define $\morph$ as $\morph(\$) = \$$, $\morph(a) = ab$, and $\morph(b) = a$. Then $\morph(f_{2k}^\$) = f_{2k+1}^\$$. It is known that  $r(f_{2k+1}^\$)/r(f_{2k}^\$) = \Omega(\log n)$ \cite{Giuliani2025}. Hence, $MS_\mu(n) = \Omega(\log n)$.
\end{example}

We first show that $MS_{\rho_{p>1}}(n)$ is bounded.

\begin{lemma}\label{le:rho_p_bounded_MS}
Let $w\in\{a,b\}^*$ be a word that contains at least two $a$'s and one $b$. Let $t$ be the length of the longest \emph{circular run} of $b$'s in $w$ (i.e., the longest run of $b$'s in any string in $\R{w}$). 
It holds that $$r(\rho_p(w)) \le r(w) + 2\left|\cfact(w)  \cap \bigcup \{ab^ia \,|\,i \in [1,t]\}\right|.$$
Moreover, it holds $\Delta_{\rho_p}^+(w) \le 2r(w)$ and $\Delta_{\rho_p}^\times(w) \le 3$.
\end{lemma}

\begin{proof}
Let $t$ be the length of the maximal circular run of $b$'s in $w$. Since $\rho_p = (a,b^p)$ is order-preserving, the sequence obtained by taking the last character of each image of the lexicographically sorted rotations of $w$ spells exactly $\bwt(w)$. In fact, $a$ and $b$ are the last characters of $\rho_p(a)$ and $\rho_p(b)$. respectively.
Hence, the last characters of the range of rotations starting with $a$ in the BWT matrix of $\rho_p(w)$ spell exactly $\bwt(w)[1,|w|_a]$. Similarly, the last characters of the (disjoint) ranges of rotations starting with $b^{ip}a$ for $i \in [1,t]$ spell exactly $\bwt(w)[|w|_a+1,|w|]$. Strictly in between the ranges of rotations starting with $b^{(i-1)p}a$ and $b^{ip}a$ for some $i \in [1,t]$, there is a range of rotations starting with $b^{(i-1)p+s}a$ for each $s \in [1,p-1]$, all ending with the character $b$. 
In the worst case, each of these blocks of rotations can only increase the number of runs of $r(w)$ by $2$.
Hence, the additive increase is at most 2 times the number of circular factors of the form $ab^ia$ in $w$. This proves the first claim of the proposition. 

For the second claim, observe that a change of letter occurs in correspondence of each block of rotations starting with $b^ia$, for each $i$ such that $ab^ia\in\cfact(w)$. Hence, the second claim follows because $|\cfact(w) \cap \bigcup \{ab^ia \,|\,i \in [1,t]\}| \le r(w)$.
\end{proof}

We now give a sketch of the main result of this section. The complete proof will be deferred to the full version of this article.

\begin{theorem}\label{thm:multiplicative_sensitivity}
For every morphism $\mu: \{a,b\}^*\rightarrow \Gamma^*$, there exists an integer $k_\mu$ such that $MS_\mu(n) \le k_\mu$.
\end{theorem}

\begin{proof}[Proof sketch of Theorem~\ref{thm:multiplicative_sensitivity}]
We assume $\morph$ is injective, as otherwise the result follows from Lemma~\ref{le:cyclic_sensitivity}. By Proposition~\ref{prop:unbounded_increase_wk_power_morphisms}, $\morph$ can be decomposed as $\mu = \eta \circ \rho_q \circ E \circ \rho_p \circ E$ with $\eta = (u, v)$ and $u,v \in Q(\Sigma^*)$. By Lemma \ref{le:rho_p_bounded_MS},  both $MS_{\rho_p}(n)\le 3$ and $MS_{\rho_q}(n)\le 3$, hence $MS_\morph(n)$ is bounded if and only if $MS_\eta(n)$ is bounded. If $\eta$ is primitivity-preserving, then by Lemma \ref{lem:primitivitypreserving->BWTrunpreserving} we are done. Hence, we are left to show the case when $\eta$ is not primitivity-preserving and both images are primitive. We give a sketch for this case.

Let $\mu =(u,v)$ be a non-primitivity-preserving injective morphism with $u,v \in Q(\Sigma^*)$. 
By Lemma~\ref{u*v*}, there exists a primitive word $x$ with $|x|>1$, such that $P^\mu=\R{x}$ and $\morph(x) = z^t$ with $z \in Q(\Sigma^*)$ and $t > 1$.

As a consequence of Lemma~\ref{le:non_prim_sync_pair}, there exists an integer $k>0$, which depends only on $\morph$, such that every rotation with a $k$-length prefix $y\notin\Gamma^k \cap \cfact(\{z\}^*)$ contains a synchronization pair. Hence, we can partition these rotations according to their length-$k$ prefix, and the characters preceding these rotations can be determined. 



The remaining rotations starting with a power of some rotation of $z$ are handled in a similar (though more complicated) fashion with respect to how rotations starting with a power of $b$ were handled in Lemma~\ref{le:rho_p_bounded_MS}. This yields an upper-bound for $MS_\mu$ depending on the value $|z|$ instead of $3$.
\end{proof}

\section{Conclusions and future work}\label{sec:final}

In this paper, we have provided a complete characterization of binary injective morphisms that preserve the number of BWT-runs up to a bounded additive increase.  We have shown that this class coincides with the class of binary primitivity-preserving morphisms. 

Primitivity-preserving morphisms could be considered a general effective tool for studying and evaluating repetitiveness measures, since such measures remain invariant, up to small constants, when applied to powers of a word. This suggests that such morphisms could be seen as a unifying framework for the analysis and comparison of different repetitiveness measures. 

It would be interesting to explore the design of compression and indexing techniques based on BWT-runs that operate directly on morphic encodings of highly repetitive text collections. This could have applications, for example, in the domain of privacy-preserving algorithms. Although our current approach allows for polynomial-time decision procedures for testing whether a given binary morphism is BWT-run preserving or, equivalently, primitivity-preserving, more efficient algorithms could yield significant improvements in terms of scalability and practical performance.

Furthermore, BWT-run sensitivity could support a new classification of morphisms, providing new insights for their structural behavior and the impact on repetitiveness measures.

Finally, we plan to investigate how to extend our results to morphisms over larger alphabets. 



\newpage
\appendix

\section{The theorem of Lyndon and Sch\"utzenberger and binary injective morphisms}
\label{ap:Lyndon}
Here, we report two classical results that are used in this paper to prove combinatorial properties of injective morphisms, and more specifically, of primitivity-preserving morphisms.

\begin{lemma}[\cite{lyndon}]
\label{commutepower}
Two words $u,v\in\Sigma^+$ commute if and only if there exist two integers $\ell,m\geq1$ and a word $z\in\Sigma^+$ such that $u=z^\ell$ and $v=z^m$.
\end{lemma}

More generally, the theorem of Lyndon and Sch\"utzenberger states that the word equation $u^\ell v^m=z^n$ has only trivial (i.e., periodic) solutions for $\ell,m,n\geq2$:

\begin{theorem}[\cite{lyndon}]
\label{lyndonS}
    Let $u,v,z\in\Sigma^+$ and $\ell, m, n\geq2$ such that $u^\ell v^m = z^n$. Then, there exists, and is unique, a primitive word $w$ such that $u=w^{t_1}$, $v = w^{t_2}$, and $z = w^{t_3}$, for some integers $t_1,t_2,t_3\geq1$.  
\end{theorem}

Notice that in the equation of Theorem~\ref{lyndonS} we can suppose, without loss of generality, that $|u|\geq|v|$, since $u^\ell v^m = z^n$ if and only if $v^m u^\ell = (z')^n$ for a rotation $z'$ of $z$. 
However, there can be nontrivial solutions when $\ell=1$ or $m=1$, and $n>1$. As an example, take $u=abba$, $v=b$. Then $uv^2=abba\cdot b \cdot b = (abb)^2$. In other words, the equation $u^\ell v=z^n$ (or $uv^m=z^n$) can have nontrivial solutions.

We now give the proof of Lemma \ref{le:P_mu_structure} in which we characterize the structure of the set of $\morph$--power words of an injective morphism $\morph$.
\begin{proof}[Proof of Lemma~\ref{le:P_mu_structure}]
Case~\ref{Pmu_1a} coincides with the definition of primitivity-preserving morphism.
Case~\ref{Pmu_1b} holds when only one between $u$ and $v$ is not primitive.
Case~\ref{Pmu_1c} holds when both $u$ and $v$ are not primitive.

Cases~\ref{Pmu_2a} and~\ref{Pmu_2b} both follow by Lemma~\ref{u*v*}, where we distinguish the case when either both $u$ and $v$ are primitive or only one between $u$ and $v$ is not primitive.
We now prove that there can not be further cases, that is if we fall in case~\ref{Pmu_2} then either $u$ or $v$ must be primitive.
By contradiction, let us assume that there exist $p,q\in\{a,b\}^*$ and $s,t >1$ such that $u = p^s$ and $v = q^t$, and let $m,n\geq1$ such that $w = a^m b^n$, and therefore $\morph(w) = u^m v^n = p^{ms} q^{nt} = z^k\in W\cap \Pow{\{a,b\}^*}$, for some primitive word $z$ and $k>1$. By Lemma~\ref{lyndonS}, it follows that $p$ and $q$ are powers of $z$, and by transitive relation so do $u$ and $v$, in contradiction with $\morph$ being an injective morphism.
\end{proof}

The following lemma reformulates~\cite[Theorem 8]{HolubRS23} and provides a parametric solution to Theorem~\ref{lyndonS}. It characterizes the structure of binary injective morphisms that map primitive words of length greater than $2$ to non-primitive words.


\begin{lemma}
\label{le:Holub_classification}
    Let $\morph=(u,v)$ be an injective morphism for some words $u,v\in\{a,b\}^*$ with $|u|\geq|v|$, and let $W = \{u^nv \mid n\geq1\} \cup \{uv^n \mid n>1\}$. If $W\cap \Pow{\{a,b\}^*}\neq \emptyset$, then exactly one of the following cases occurs:
    \begin{enumerate}
        \item $u= (pq)^m p$ and $v= q(pq)^n$, for some non-commuting words $p,q$ and two integers $m, n\geq0$ such that $m+n\geq1$. In this case, $W\cap \Pow{\{a,b\}^*} = \{uv\}$ and $\R{ab}\subseteq P^\morph$;
        \item $u=(pq^{n})^m p$ and $ v=q$, for some non-commuting words $p,q$ and three integers $m,n\geq1$.  In this case, $W\cap \Pow{\{a,b\}^*} = \{uv^n\}$ and $\R{ab^n}\subseteq P^\morph$;
        \item $u = pq(q(pq)^m)^{k-1})^n pq(q(pq)^m)^{k-2}qp$ and $v= q(pq)^m$, for some non-commuting words $p,q$ and three integers $k\geq2$, $m\geq1$, and $n\geq0$.  In this case, $W\cap \Pow{\{a,b\}^*} = \{uv^k\}$ and $\R{ab^k}\subseteq P^\morph$;
        \item $u =(pq)^m p$ and $v= qppq$, for some non-commuting words $p,q$ and an integer $m\geq2$.  In this case, $W\cap \Pow{\{a,b\}^*} = \{u^2v\}$ and $\R{a^2b}\subseteq P^\morph$.
    \end{enumerate}
\end{lemma}

Note that the hypothesis $|\morph(a)| \geq |\morph(b)|$ in Lemma~\ref{le:Holub_classification} is not restrictive, as when $|\morph(a)| < |\morph(b)|$ we can consider the morphism $\morph$ composed with the morphism $E$ that exchanges $a$ and $b$.

We give now the proof of Lemma \ref{prop:purenotcircular=rotation} in which we state that a binary primitivity-preserving morphism is not recognizable if and only if $\morph(a)$ and $\morph(b)$ are not conjugates.

\begin{proof}[Proof of Lemma \ref{prop:purenotcircular=rotation}]
    Without loss of generality, we can suppose throughout the proof that $|u|\geq|v|$.
    Observe that the statement is equivalent to: $\morph$ is not recognizable if and only if $\morph=(pq,qp)$ for some words $p,q$. 
    
    Let us prove the first direction by contraposition. If $\morph$ is injective and $u$ and $v$ are conjugates, then there exist two non-empty words $p,q$ such that $u=pq$ and $v=qp$.
    It is easy to see then that $pq,qp\in \{u,v\}^+$ but $p,q\notin \{u,v\}^+$, which by Lemma~\ref{le:rec=circ} implies that $\morph$ is not recognizable. 

    For the other direction, also by contraposition, if $\morph$ is not circular, then there exist $k,k'\geq0$ and $w\in \{u,v\}^+$ such that $w = x_1x_2\cdots x_k = q' y_2 \cdots y_{k'} p'$ with $y_1=p'q'$, $p',q'\neq \varepsilon$, and $x_1,x_2,\ldots,x_{k}, y_1,y_2,\ldots,y_{k'}\in \{u,v\}$.
    We now show that, under the above-mentioned conditions, if $u$ and $v$ are not conjugated, the morphism $\morph$ can not be primitivity-preserving, leading to a contradiction.
    
    Observe that both $u$ and $v$ have to occur in $w$ circularly in both the factorizations, otherwise, we end up having a prefix and a suffix of $u$ (or $v$) that commute, which by Lemma~\ref{commutepower} implies that $u$ (or $v$) is a power, contradicting the hypothesis of $\morph$ being primitivity-preserving.
    Let $X_\morph=\{u,v\}$.
    We can then distinguish three cases: (i) $u^2$ is a $X_\morph$-factor of $w$, (ii) $u^2$ is not a $X_\morph$-factor of $w$ and $uv^\ell u$ is, for some $\ell>0$, and (iii) neither $u^2$ nor $uv^\ell u$, for all $\ell>0$, are $X_\morph$-factors of $w$ but $uv^m$ is, for some $m>0$.
    For case (i), by~\cite[Proposition A]{RestR85} follows that $\morph(a^2b) = u^2 v$ is a power, and by Lemma~\ref{le:pure=primitive} this contradicts $\morph$ being primitivity-preserving.
    For case (ii), if $uv^\ell u$ is a $X_\morph$-factor of $w$ for some $\ell>0$ and $u^2$ is not, then $w\in uv^+(uv^+)^+$, and therefore $|w|\geq2|u|+2|v|$. By~\cite[Proposition B]{RestR85} follows that $\morph(ab^m) = uv^m$ is a power for some $m>0$, which again by Lemma~\ref{le:pure=primitive} it contradicts $\morph$ being primitivity-preserving.
    Finally, for case (iii), observe that if neither $u^2$ nor $uv^\ell u$ are $X_\morph$-factors of $w$ for all $\ell>0$, then there exist $i\in[1,k],j\in[1,k']$ such that $x_i,y_j=u$ and $x_{i'}=y_{j'}=v$ for all $i'\neq i$, $j'\neq j$. This implies that $k=k'>1$ and that two rotation words coincide, and by Lemma~\ref{commutepower} follows that $\morph(ab^{k-1})=uv^{k-1}$ is a power, i.e. $\morph$ is not primitivity-preserving, and the thesis follows.
\end{proof}

The following result is used to show that it is possible to test in polynomial time, with respect to the total length of the images of the letters, whether a morphism is BWT-run preserving.

\begin{theorem}[\cite{HolubRS23}]
\label{th:holub}
 Let $\morph=\{u,v\}$ be an injective morphism, with $|u|\geq |v|$, and let $W = \{u^nv \mid n\geq1\} \cup \{uv^n \mid n>1\}$. If there exists a primitive word $z$ such that $u^\ell v^m = z^n$ for some $\ell,m\geq1$, $n>1$, then: 
 \begin{enumerate}
     \item if $\ell>1$,  $\ell = n = 2$ and $m=1$;
     \item if $\ell = 1$,  $1\leq m \leq \frac{|u|-4}{|v|}+2$.
 \end{enumerate}
\end{theorem}

\section{Proof of Lemma~\ref{le:recmorphdelay->BWTrunbounded}}
\label{ap:l32}

A morphism $\morph=(u,v)$ is called \emph{prefix} (resp. \emph{suffix}) if neither $u$ is a prefix (resp. suffix) of $v$ nor $v$ is a prefix (resp. suffix) of $u$.
Additionaly, $\morph$ is called \emph{bifix} if it is both prefix and suffix.
We first prove some properties used in the proof.

\begin{lemma}
\label{prefix} \label{le:prefix}
    Let $\morph:\{a,b\}^*\rightarrow\Gamma^*$ be an injective morphism, and let $\fibmorph = (ab,a)$ and $E=(b,a)$.
    Then, the morphism $\morph$ is prefix if and only if for all $\psi:\{a,b\}^*\rightarrow\Gamma^*$ and $\chi:\{a,b\}^*\rightarrow\{a,b\}^*$ such that $\morph = \psi \circ \chi$, it holds $\chi\notin\{\varphi,\varphi\circ E\}$.
\end{lemma}

\begin{proof}
    For the first direction, suppose by contradiction that exists $\psi$ such that either $\morph = \psi \circ \varphi$ or $\morph = \psi \circ \varphi \circ E$.
    By construction, we obtain either $\morph = (\psi(a)\psi(b), \psi(a))$ or $\morph = (\psi(a), \psi(a)\psi(b))$, contradiction.

    The other direction follows by construction.
\end{proof}

Using symmetrical arguments, we obtain the following lemma. 

\begin{lemma}
\label{le: suffix}
   Let $\morph:\{a,b\}^*\rightarrow\Gamma^*$ be an injective morphism, and let $\Tilde{\fibmorph} = (ba,a)$ and $E=(b,a)$.
    Then, the morphism $\morph$ is suffix if and only if for all $\psi:\{a,b\}^*\rightarrow\Gamma^*$ and $\chi:\{a,b\}^*\rightarrow\{a,b\}^*$ such that $\morph = \psi \circ \chi$, it holds $\chi\notin\{\Tilde{\fibmorph},\Tilde{\fibmorph}\circ E\}$.
\end{lemma}

From Lemmas \ref{le:prefix} and \ref{le: suffix}, the following proposition can be derived. 
 \begin{proposition}
\label{prop:structbifix}
    Let $\morph:\{a,b\}^*\rightarrow\Gamma^*$ be an injective morphism, and let $\fibmorph = (ab,a)$, $\Tilde{\fibmorph} = (ba,a)$, and $E=(b,a)$.
    Then, the morphism $\morph$ is bifix if and only if for all $\psi:\{a,b\}^*\rightarrow\Gamma^*$ and $\chi:\{a,b\}^*\rightarrow\{a,b\}^*$ such that $\morph = \psi \circ \chi$, it holds $\chi\notin\{\fibmorph,\fibmorph\circ E, \Tilde{\fibmorph},\Tilde{\fibmorph}\circ E\}$.
\end{proposition}

We now give the proof of Lemma \ref{le:recmorphdelay->BWTrunbounded}, where we state that the finite-delay synchronization of a morphism on the images of a language results in a bounded increase in the number of BWT-runs

\begin{proof}[Proof of Lemma~\ref{le:recmorphdelay->BWTrunbounded}]
    Let $u\in \La$ and let $w=\morph(u)$. We denote by $w_i$ the $i$th rotation in lexicographic order, for all $i\in[0,n)$, where $n=|w|$.
    Let $\cfact_k = \cfact(\morph(\La)) \cap \Gamma^k$, and let $m=|\cfact_k|$. 
    We denote by $f_j\in \cfact_k$ the $j$th factor in lexicographic order, for all $j\in[0,m)$.
    We can then partition the set $\R{w}=\{w_0,\ldots w_{n-1}\}$ into a finite number $m$ of subsets $R_0,\ldots,R_{m-1}$ such that $R_j = \{w_i\mid w_i[0,k-1]=f_j\}$, for all $j\in[0,m)$.
    Observe that for each $j$ there exist $\min_j = \min\{i\mid w_i\in R_j\}$ and $\max_j = \max\{i\mid w_i\in R_j\}$ such that $R_j=\{w_i\mid i\in[\min_j,\max_j]\}$.
    Let us suppose $\morph$ is bifix, and let $\ell = lcs(\morph(a), \morph(b))$.
    Since $\morph$ is synchronizing with delay $k$ on $\morph(\La)$, this implies that for all $j\in[0,m)$ there exists a syncronization point in $f_j = p_jv_js_j$, for some $v_j\in\morph(\Sigma^*)$ and $p_j,s_j\in\Gamma^*$ such that $p_j$ and $s_j$ are a proper suffix and a proper prefix respectively of either $\morph(a)$ or $\morph(b)$.
    If $0<|p_j|<\ell$, i.e. $p_j$ is a proper suffix of the longest common suffix between $\morph(a)$ and $\morph(b)$, then $\bwt[i]=\morph(a)[|\morph(a)|-|p_j|-1] = \morph(b)[|\morph(b)|-|p_j|-1]$, for all $i\in[\min_j,\max_j]$.
    If $|p_j|>|\ell|$, then $p_j$ is either a proper suffix of $\morph(a)$, and therefore $\bwt[i]=\morph(a)[|\morph(a)|-|p_j|-1]$, or a proper suffix of $\morph(b)$, and therefore $\bwt[i]=\morph(b)[|\morph(b)|-|p_j|-1]$, for all $i\in [\min_j,\max_j]$.
    If $p_j = \ell$, then $w_i = \morph(u')[n-|p_j|,n-1]\cdot \morph(u')[0,n-|p_j|-1]$,  for all $i\in [\min_j,\max_j]$ and for some $u'\in\R{u}$.
    Let $J = \bigcup \{j\mid |p_j|=\ell\}$. It is easy to see that $|\bigcup_{j\in J}R_j|=|u|$, and we write $j_i$ to denote the $i$th element in $J$ in increasing order. By~\cite[Lemma~11]{Fici23} it follows that for each pair $u',u''\in \R{u}$, either $u'<u'' \iff \morph(u')<\morph(u'')$, or $u'<u'' \iff \morph(u')>\morph(u'')$. Observe that for any pair of words $w_1,w_2\in\Gamma^*$ and letter $c\in\Gamma$, we have $w_1c < w_2c \iff cw_1 <cw_2$. Hence, we can conclude that $\bwt[\min_{j_1},\max_{j_1}]\cdots\bwt[\min_{j_{|J|}},\max_{j_{|J|}}]$ spells $\bwt(u)$, up to reverse operation and/or exchanging $a$'s and $b$'s with letters from $\Gamma$.
    On top of these $r(u)$ BWT-runs, we have to count that each of the $m-|J|$ BWT-runs in correspondence of the range $[\min_j,\max_j]$ of rotations such that $j\notin J$ can increase the number of BWT-runs by at most 2, it follows that $r(w) \leq r(u)+2(m-|J|)\leq r(u)+2m$. Since $m$ is finite, the thesis follows.
    
     Let $\Fib = \{\fibmorph,\fibmorph\circ E,\Tilde{\fibmorph},\Tilde{\varphi}\circ E\}$. If $\morph$ is not bifix, by Proposition~\ref{prop:structbifix} we can write $\morph = \eta\circ \psi_1\circ\cdots \circ \psi_t$ such that $\eta\notin\Fib$ and $\psi_1,\ldots,\psi_t\in \Fib$.
     Since $\psi_1, \ldots, \psi_t$ are recognizable~\cite{codesautomata}, it follows that $\morph$ is synchronizing with delay $k$ on $\morph(\La)$ if and only if $\eta$ is synchronizing with delay $k'\leq k$ on $\eta(\psi_1\circ\cdots\circ\psi_t(\La))=\morph(\La)$.
     Moreover, by~\cite[Theorem~21]{Fici23}, we have that $r(u)=r(\psi_1\circ \cdots \circ \psi_t(u))$; hence, we can show the proof for the bifix morphism $\eta$, and the thesis follows.
\end{proof} 

\section{Proofs of Propositions \ref{prop:rho_unbounded_AS} and \ref{prop:unbounded_increase_wk_power_morphisms}}
\label{ap:prop3435}

\begin{definition}[\cite{Giuliani2025}]\label{def:w_k}
For every $k > 5$,  let $s_i = ab^iaa$ and $e_i = ab^iaba^{i-2}$ for all $2 \leq i \leq k-1$, and $q_k = ab^ka$.
We define the word $$w_k = \left(\prod_{i=2}^{k-1}s_i e_i\right)q_k = \left(\prod_{i=2}^{k-1}ab^iaa ab^iaba^{i-2}\right)ab^ka.$$ 
\end{definition}

\begin{proof}[Proof of Proposition~\ref{prop:rho_unbounded_AS}]
Let us consider the family of words $w_k$ from Definition~\ref{def:w_k}.
First note that $\rho_p = (a, b^p)$ is abelian order-preserving. This implies that the last letters of the range of rotations starting with $a$ in the BWT matrix of $\rho_p(w_k)$ spell exactly $\bwt(w_k)[1,|w_k|_a]$. Similarly, the last characters of the (disjoint) ranges of rotations starting with $b^{ip}a$ for $i \in [1,k]$ spell exactly $\bwt(w_k)[|w_k|_a+1,|w_k|]$. Moreover, it has been shown in \cite{Giuliani2025} that the blocks of rotations of $w_k$ starting with $b^ia$, for some $i \in [1, k]$, spell a substring of $\bwt(w_k)$ that begins and ends with the letter $a$. Hence, the same holds for the blocks $b^{ip}a$ in $\rho_p(w_k)$. Strictly in between the ranges of rotations starting with $b^{(i-1)p}a$ and $b^{ip}a$ for some $i \in [2,k]$, there is a range of rotations starting with $b^{(i-1)p+s}a$ for each $s \in [1,p-1]$, all ending with the character $b$. These new blocks of rotations increase the number of runs by $2$ each,  and there are $k-1$ of them. Since $k=\Theta(\sqrt{n})$, the claim holds.
\end{proof}

\begin{proof}[Proof of Proposition~\ref{prop:unbounded_increase_wk_power_morphisms}]
The composed morphism maps $a$ to $u^p$ and $b$ to $v^q$ through the substitution chains

$$a \xrightarrow{E} b  \xrightarrow{\rho_p} b^p \xrightarrow{E} a^p \xrightarrow{\rho_q} a^p \xrightarrow{\eta} u^p \text{ and } b \xrightarrow{E} a \xrightarrow{\rho_p} a \xrightarrow{E} b \xrightarrow{\rho_q} b^q \xrightarrow{\eta} v^q.$$

For the second claim, note that $AS_{\mu}(n) \ge AS_{\rho_p}(n)$, as $$r(\mu(E(w))) - r(E(w)) = r(\eta \circ \rho_q \circ E \circ \rho_p (w)) - r(w) \ge r(\rho_p(w)) - r(w)$$ holds for any word $w$. Similarly, when $p = 1$, it holds $AS_{\mu}(n) \ge AS_{\rho_q}(n)$, as
$$r(\mu(w)) - r(w) = r(\eta \circ \rho_q(w)) - r(w) \ge r(\rho_p(w)) - r(w).$$ 
By Proposition \ref{prop:rho_unbounded_AS}, when $pq > 1$, the claim follows.
\end{proof}


\begin{thebibliography}{10}

\bibitem{AKAGI2023}
Tooru Akagi, Mitsuru Funakoshi, and Shunsuke Inenaga.
\newblock Sensitivity of string compressors and repetitiveness measures.
\newblock {\em Information and Computation}, 291:104999, 2023.

\bibitem{BealPR_EJC24}
Marie{-}Pierre B{\'{e}}al, Dominique Perrin, and Antonio Restivo.
\newblock Unambiguously coded shifts.
\newblock {\em European Journal of Combinatorics}, 119:103812, 2024.

\bibitem{codesautomata}
Jean Berstel, Dominique Perrin, and Christophe Reutenauer.
\newblock {\em Codes and Automata}, volume 129 of {\em Encyclopedia of
  mathematics and its applications}.
\newblock Cambridge University Press, 2010.

\bibitem{DBLP:conf/mfcs/BerstelS93}
Jean Berstel and Patrice S{\'{e}}{\'{e}}bold.
\newblock A {C}haracterization of {S}turmian {M}orphisms.
\newblock In {\em {MFCS}}, volume 711 of {\em Lecture Notes in Computer
  Science}, pages 281--290. Springer, 1993.

\bibitem{BW94}
Michael Burrows and David Wheeler.
\newblock A block sorting lossless data compression algorithm.
\newblock Technical Report 124, Digital Equipment Corporation, 1994.

\bibitem{RomanaJCTA}
Julien Cassaigne, France Gheeraert, Antonio Restivo, Giuseppe Romana, Marinella
  Sciortino, and Manon Stipulanti.
\newblock New string attractor-based complexities for infinite words.
\newblock {\em Journal of Combinatorial Theory, Series {A}}, 208:105936, 2024.

\bibitem{DBLP:journals/tcs/Chuan99}
Wai{-}Fong Chuan.
\newblock Sturmian morphisms and alpha-words.
\newblock {\em Theoretical Computer Science}, 225(1-2):129--148, 1999.

\bibitem{domosi_93}
P{\'a}l D{\"o}m{\"o}si, S{\'a}ndor Horv{\'a}th, Masami Ito, L{\'a}szl{\'o}
  K{\'a}szonyi, and Masashi Katsura.
\newblock Formal languages consisting of primitive words.
\newblock In Zolt{\'a}n {\'E}sik, editor, {\em Fundamentals of Computation
  Theory}, pages 194--203, Berlin, Heidelberg, 1993. Springer Berlin
  Heidelberg.

\bibitem{Fici23}
Gabriele Fici, Giuseppe Romana, Marinella Sciortino, and Cristian Urbina.
\newblock {On the Impact of Morphisms on BWT-Runs}.
\newblock In {\em {CPM}}, volume 259 of {\em LIPIcs}, pages 10:1--10:18.
  Schloss Dagstuhl - Leibniz-Zentrum f{\"{u}}r Informatik, 2023.

\bibitem{rindex}
Travis Gagie, Gonzalo Navarro, and Nicola Prezza.
\newblock {Fully Functional Suffix Trees and Optimal Text Searching in BWT-Runs
  Bounded Space}.
\newblock {\em Journal of the {ACM}}, 67(1):2:1--2:54, 2020.

\bibitem{Giuliani2025}
Sara Giuliani, Shunsuke Inenaga, Zsuzsanna Lipt{\'a}k, Giuseppe Romana,
  Marinella Sciortino, and Cristian Urbina.
\newblock Bit catastrophes for the {B}urrows-{W}heeler transform.
\newblock {\em Theory of Computing Systems}, 69(19), 2025.

\bibitem{HolubRS23}
Stepan Holub, Martin Raska, and Step{\'{a}}n Starosta.
\newblock Binary codes that do not preserve primitivity.
\newblock {\em Journal of Automated Reasoning}, 67(3):25, 2023.

\bibitem{Huang}
Cheng{-}Chi Huang.
\newblock A note on pure codes.
\newblock {\em Acta Informatica}, 47(5-6):347--357, 2010.

\bibitem{KempaP18}
Dominik Kempa and Nicola Prezza.
\newblock At the roots of dictionary compression: string attractors.
\newblock In {\em {STOC}}, pages 827--840. {ACM}, 2018.

\bibitem{bowtie2}
Ben Langmead and Steven~L Salzberg.
\newblock Fast gapped-read alignment with {B}owtie 2.
\newblock {\em Nature Methods}, 9(4):357--359, 2012.

\bibitem{Bowtie}
Ben Langmead, Cole Trapnell, Mihai Pop, and Steven~L Salzberg.
\newblock Ultrafast and memory-efficient alignment of short {DNA} sequences to
  the human genome.
\newblock {\em Genome Biology}, 10(3):R25, 2009.

\bibitem{bwa}
Heng Li and Richard Durbin.
\newblock Fast and accurate long-read alignment with {Burrows-Wheeler}
  transform.
\newblock {\em Bioinformatics}, 26(5):589--595, 2010.

\bibitem{lothaire1997combinatorics}
M.~Lothaire.
\newblock {\em Combinatorics on Words}.
\newblock Cambridge University Press, 1997.

\bibitem{lyndon}
Roger~C. Lyndon and Marcel-Paul Sch{\"u}tzenberger.
\newblock The equation $a^m= b^n c^p$ in a free group.
\newblock {\em Michigan Mathematical Journal}, 9(4):289--298, 1962.

\bibitem{DBLP:journals/ipl/MantaciRS03}
Sabrina Mantaci, Antonio Restivo, and Marinella Sciortino.
\newblock Burrows--{W}heeler transform and {S}turmian words.
\newblock {\em Information Processing Letters}, 86(5):241--246, 2003.

\bibitem{Mitrana97}
Victor Mitrana.
\newblock Primitive morphisms.
\newblock {\em Information Processing Letters}, 64(6):277--281, 1997.

\bibitem{Navacmcs20.3}
Gonzalo Navarro.
\newblock Indexing highly repetitive string collections, part {I}:
  Repetitiveness measures.
\newblock {\em ACM Computing Surveys}, 54(2):article 29, 2021.

\bibitem{Navarro22cacm}
Gonzalo Navarro.
\newblock The compression power of the {BWT:} technical perspective.
\newblock {\em Communications of the {ACM}}, 65(6):90, 2022.

\bibitem{NAVARRO_Urbina_TCS2025}
Gonzalo Navarro and Cristian Urbina.
\newblock Repetitiveness measures based on string morphisms.
\newblock {\em Theoretical Computer Science}, page 115259, 2025.
\newblock In press.

\bibitem{DBLP:journals/tcs/Paquin09}
Geneviève Paquin.
\newblock On a generalization of {C}hristoffel words: epichristoffel words.
\newblock {\em Theoretical Computer Science}, 410(38-40):3782--3791, 2009.

\bibitem{RestR85}
Evelyne~Barbin{-}Le Rest and Michel~Le Rest.
\newblock Sur la combinatoire des codes {\`{a}} deux mots.
\newblock {\em Theoretical Computer Science}, 41:61--80, 1985.

\bibitem{Restivo74}
Antonio Restivo.
\newblock {On a Question of McNaughton and Papert}.
\newblock {\em Information and Control}, 25(1):93--101, 1974.

\bibitem{Rigo2014}
Michel Rigo.
\newblock {\em Formal Languages, Automata and Numeration Systems 1:
  Introduction to Combinatorics on Words}.
\newblock Wiley, 2014.

\bibitem{Shyr77}
Huei{-}Jan Shyr and Gabriel Thierrin.
\newblock Codes, languages and {MOL} schemes.
\newblock {\em {RAIRO} Theoretical Informatics and Applications},
  11(4):293--301, 1977.

\bibitem{ShyrThierrin1997}
Huei{-}Jan Shyr and Gabriel Thierrin.
\newblock Codes, languages and {MOL} schemes.
\newblock {\em {RAIRO} Theoretical Informatics and Applications / Informatique
  Théorique et Applications}, 31(4):293--301, 1997.
\newblock \href {https://doi.org/10.1051/ita:1997131}
  {\path{doi:10.1051/ita:1997131}}.

\bibitem{ShyrYu}
Huei-Jan Shyr and Shyr-Shen Yu.
\newblock Non-primitive words in the language $p^+q^+$.
\newblock {\em Soochow Journal of Mathematics}, 20:535–546, 1994.

\bibitem{VasimuddinMLA19}
Md. Vasimuddin, Sanchit Misra, Heng Li, and Srinivas Aluru.
\newblock Efficient architecture-aware acceleration of {BWA-MEM} for multicore
  systems.
\newblock In {\em 2019 {IEEE} International Parallel and Distributed Processing
  Symposium, {IPDPS} 2019, Rio de Janeiro, Brazil, May 20-24, 2019}, pages
  314--324, Los Alamitos, CA, 2019. {IEEE} Computer Society.

\end{thebibliography}
\end{document}